\documentclass[letter,12pt]{article}
\usepackage{amsmath,amssymb,amsthm,mathtools}
\usepackage{xcolor}
\usepackage{graphicx}
\usepackage[authoryear,round,longnamesfirst]{natbib}
\usepackage[colorlinks,citecolor={blue}]{hyperref}
\usepackage{doi}
\usepackage{xurl} % break long urls

\usepackage[inline,shortlabels]{enumitem} % for inline enumeration
\setlist[enumerate,1]{label=(\roman*)}
\usepackage{subcaption} % for subfigures
\usepackage{booktabs} % for nicer tables
\usepackage{xspace}

\newcommand{\R}{\mathbb{R}} % real numbers
\newcommand{\E}{\operatorname{\mathbb{E}}}
\newcommand{\Var}{\operatorname{Var}}
\newcommand{\by}{\mathbf{y}}
\newcommand{\bV}{\mathbf{V}}
\newcommand{\bY}{\mathbf{Y}}
\newcommand{\iid}{\textsc{iid}\xspace} % we can change iid symbol to whatever
\newcommand{\abs}[1]{\left\lvert #1 \right\rvert}

\renewcommand{\hat}{\widehat} % \widehat is nicer
\newcommand{\set}[1]{\left\{ #1 \right\}}
\newcommand{\diff}{\mathop{}\!\mathrm{d}}
\newcommand{\e}{\mathrm{e}}
\newcommand{\cI}{\mathcal{I}}
\newcommand{\cN}{\mathcal{N}}
\newcommand{\cT}{\mathcal{T}}
\newcommand{\cW}{\mathcal{W}}
\newcommand{\diag}{\operatorname{diag}} % diagonal matrix

\newcommand{\pto}{\xrightarrow{p}}

\newtheorem{theorem}{Theorem}

\newtheorem{algorithm}{Algorithm}

\newtheorem{condition}{Condition}

\newtheorem{corollary}{Corollary}

\newtheorem{definition}{Definition}

\newtheorem{lemma}{Lemma}

\newtheorem{proposition}{Proposition}

\newtheorem{asmp}{Assumption}
\allowdisplaybreaks
\usepackage[margin = 1.25in]{geometry}
\usepackage{setspace}
\begin{document}

\title{\Large Fixed-$k$ Tail Regression: New Evidence on\\ Tax and Wealth Inequality from Forbes 400}
%\title{\Large The Top Marginal Tax Rate and the Wealth Inequality: \\Evidence from Forbes 400} How about this title?
%\title{\Large Quantifying the Effect of Top Income Tax Rate on Wealth Inequality: A New Method with Evidence from Forbes 400} %Fine for econometrics papers, but I prefer the following title for AER. 
%\title{\Large Tax Progressivity and Wealth Inequality: \\Evidence from Forbes 400}
\author{\normalsize
Ji Hyung Lee\thanks{\footnotesize\setlength{\baselineskip}{4.4mm} 
Ji Hyung Lee: \href{mailto:jihyung@illinois.edu}{jihyung@illinois.edu}. Department of Economics, University of Illinois, 214 David Kinley Hall, 1407 West Gregory Dr, Urbana, IL 61801, USA}
\and\normalsize
Yuya Sasaki\thanks{\footnotesize\setlength{\baselineskip}{4.4mm} Yuya Sasaki: \href{mailto:yuya.sasaki@vanderbilt.edu}{yuya.sasaki@vanderbilt.edu}. Department of Economics, Vanderbilt University, VU Station B \#351819, 2301 Vanderbilt Pl, Nashville, TN 37235-1819, USA\smallskip}
\and\normalsize
Alexis Akira Toda\thanks{\footnotesize\setlength{\baselineskip}{4.4mm} Alexis Akira Toda: \href{mailto:atoda@ucsd.edu}{atoda@ucsd.edu}. Department of Economics, University of California San Diego, 9500 Gilman Dr, \#0508, La Jolla, CA 92093-0508, USA\smallskip}
\and\normalsize
Yulong Wang\thanks{\footnotesize\setlength{\baselineskip}{4.4mm} Yulong Wang: \href{mailto:ywang402@syr.edu}{ywang402@syr.edu}. Department of Economics, Syracuse University, 110 Eggers Hall, Syracuse, NY 13244-1020, USA}
}
\date{}
\maketitle

\begin{abstract}
%Using Forbes 400 data together with historical data on tax rates and macroeconomic indicators, we study the effects of the maximum marginal income tax rate on wealth inequality.
%To this end, we develop a novel fixed-$k$ tail regression method that accommodates the unique feature that the data are truncated from below at the 400th largest order statistic.
We develop a novel fixed-$k$ tail regression method that accommodates the unique feature in the Forbes 400 data that observations are truncated from below at the 400th largest order statistic.
Applying this method, we find that higher maximum marginal income tax rates induce higher wealth Pareto exponents. 
Setting the maximum tax rate to 30--40\% (as in U.S. currently) leads to a Pareto exponent of 1.5--1.8, while counterfactually setting it to 80\% (as suggested by \citealp{piketty2013capital}) would lead to a Pareto exponent of 2.6. 
We present a simple economic model that explains these findings and discuss the welfare implications of taxation. 
%Data used to study wealth inequality are often bottom-truncated.
%\textcolor{red}{Abstract needs to be updated}

%This paper proposes a novel tail regression method to estimate conditional tail index models with bottom-truncated data. 
%Unlike existing methods, our proposed method enjoys
%\begin{enumerate*}
%\item no parametric assumption on the underlying distribution,
%\item robustness against truncation and dependence among order statistics, and 
%\item validity under time-series dependence of macroeconomic control variables.
%\end{enumerate*}
%Applying it to Forbes 400 data, which is bottom-truncated at the 400th order statistic, we find that the maximum marginal income tax rates are significantly associated with wealth inequality. 

%The method is based on maximizing the asymptotic likelihood function and allows for data truncation from below. 
%Motivated by the wealth inequality literature in macroeconomics, 
%Applying the proposed method to the Forbes 400 data set, we find that higher maximum marginal income tax rates induce (\edits{``are associated with'', to avoid causal statements?}) heavier tails of the wealth distribution. 
%Counterfactural analysis, a structural macroeconomic model, and a simulation study are also provided.



%\vfill 

\end{abstract}

\newpage
%%%%%%%%%%%%%%%%%%%%%%%%%%%%%%%%%%%%%%%%%%%%%%%%%%%%%%%%%%%%%%%%%%%%%%
\section{Introduction}\label{sec:introduction}
Wealth inequality and its relationship with income tax rates are central issues in public policy debates.
%The existing literature studies these issues with mainly quantitative models.
%This paper aims to provide a new econometric tool to empirically estimate the relationship between the tax rate and wealth inequality.
%\edits{Need to update this paragraph.}
Since the distribution of wealth, denoted by $Y$, exhibits a Pareto tail in data,\footnote{See, for instance, \cite{Pareto1896LaCourbe,Pareto1897Cours}, \cite{KlassBihamLevyMalcaiSolomon2006},  \cite{nirei-souma2007}, and \cite{Vermeulen2018}, among others. See \cite{gabaix2009}, \cite{IbragimovIbragimovWalden2015}, and \cite{Gabaix2016JEP} for comprehensive reviews of this literature.} 
the Pareto exponent has been the parameter of interest as a measure of inequality. A large theoretical literature has developed quantitative macroeconomic models that generate and explain Pareto tails in wealth.\footnote{See \citet{BenhabibBisinZhu2011}, \cite{Toda2014JET,Toda2019JME}, \cite{TodaWalsh2015JPE}, \citet{GabaixLasryLionsMoll2016}, \citet{nirei2016pareto}, \citet{AokiNirei2017}, \citet{CaoLuo2017}, \citet{JonesKim2018}, \cite{BenhabibBisinLuo2019}, \cite{MaStachurskiToda2020JET}, and \cite{deVriesTodaLIS}, among many others.} This literature shows that the Pareto exponent can be expressed as a function of macroeconomic fundamentals, denoted by $X$. 
However, a formal econometric method of estimation and inference for the effects of the fundamentals $X$ on the Pareto exponent, or more generally, the tail index of the distribution of $Y$, is arguably missing. 
The main contribution of this paper is to develop such an econometric method.
%Our paper contributes to this literature by developing an econometric method of estimation and inference for the effects of the fundamentals $X$ on the Pareto exponent, and more generally, the tail index of the distribution of $Y$.

Consider a repeated cross section $\set{ Y_{it} : i=1,\dots ,n, t=1,\dots ,T}$ of the values of wealth, where $i$ indexes individuals and $t$ indexes calendar years, along with a time series $\set{X_t}_{t=1}^T$ of macroeconomic fundamentals. 
%If this data set were fully observed, then the estimation problem would be relatively straightforward. 
%However, such a data set is rarely available due to confidentiality concerns. 
%Some researchers employ the Forbes 400 data that list the values of the net worth of the 400 wealthiest individuals in each year, i.e., the 400 largest order statistics of $\set{ Y_{it} : i=1,\dots ,n}$ at each $t$. 
Our paper is motivated by the Forbes 400 data set, which lists the values of the net worth for the 400 wealthiest individuals in each year, i.e., the 400 largest order statistics of $\set{ Y_{it} : i=1,\dots ,n}$ in each $t$. 
%In comparison, other commonly used data sets typically do not provide tail observations due to confidentiality concerns. 
In contrast to other commonly used data sets, which typically do not provide tail observations due to confidentiality concerns, the Forbes 400 data set is particularly suitable for studying tail inequality. 
Featured studies using this data set include \cite{KlassBihamLevyMalcaiSolomon2006}, \citet{KaplanRauh2013AER,KaplanRauh2013JEL}, \citet{Capehart2014}, \citet{GarleanuPanageas2017}, and \citet{Gomez2021}.
 
Although the Forbes data have been explored in the above literature, rigorous econometric model and method of estimation and inference for such data sets have been missing.
Researchers typically model the extreme values of wealth in the Forbes data as a random sample drawn from a Pareto distribution to conduct the maximum likelihood estimation (MLE; e.g., \citealp{Hill1975}) or through the log-log rank-size regression (e.g., \citealp{Gabaix2011rank}); see the literature review ahead for more discussions. 
However, these approaches are not suitable in our context for two reasons.
First, the existing approaches are typically based on the increasing-$k$ asymptotic framework, where $k$ defines the number of extreme values of wealth.    
Estimation of the tail index in this increasing-$k$ framework can asymptotically ignore the correlations among these tail observations. 
In contrast, $k$ is fixed at 400 in our context, and hence the correlation among the tail observations cannot be ignored even if $n$ is large. 
Second, the large sample theory typically assumes the data to be fully observed, whereas the Forbes data are truncated from below at the 400th order statistic.
To resolve these two issues, we adopt the recently developed fixed-$k$ asymptotic framework \citep[e.g.,][]{MuellerWang2017} and develop a new estimation method, which we refer to as the fixed-$k$ tail regression. 

For a fixed positive integer $k \ge 3$, our proposed method uses the largest $k$ order statistics $\set{ Y_{(i)t} : i=1,\dots ,k}$ from $\set{ Y_{it} : i=1,\dots ,n}$ along with macro explanatory variables $X_{t}$ at each $t$.
While the underlying sample size $n$ (reflecting the U.S. population) is assumed to diverge to infinity, we allow $k$ to be fixed, unlike the aforementioned existing approaches that require $k$ to diverge.
Thus, our approach models the data generating process where observations are bottom-truncated at a fixed number $k=400$ as in the Forbes data, resolving one of the two issues with existing approaches mentioned above.
We approximate the conditional joint distribution of the $k$ order statistics $\set{ Y_{(i)t} : i=1,\dots ,k}$ (after a location and scale normalization) given $X_{t}$ with a known distribution based on extreme value theory (EVT) in the limit as $n \to \infty$ for each $t$.
Our limiting conditional joint distribution accounts for statistical dependence among the $k$ order statistics, resolving the other issue with the existing MLE approaches mentioned above.
Even after allowing for such statistical dependence, our limiting distribution can still be characterized solely by the tail index $\xi$ that governs the tail heaviness of the distribution of $Y_{it}$. 
%\footnote{In the special case where $Y_{it}$ has a Pareto upper tail, $\xi$ is equal to the reciprocal of the Pareto exponent. That said, we allow the distribution of $Y_{it}$ to be nonparametric.}
Modeling $\xi$ as a function $\xi(X_{t})$ of macro control variables $X_t$, we conduct the MLE of $\xi(\cdot)$ and perform statistical inference and counterfactual analysis. 
To the best of our knowledge, this paper is the first to establish such a parametric conditional density approximation for a nonparametric class of conditional distributions -- we summarize this main auxiliary result as Lemma \ref{lemma:hellinger} in the appendix, which may be of independent interest to econometricians.

Using the proposed econometric method, we discover that the tail index or the Pareto exponent\footnote{Technically, the Pareto exponent is only defined when the underlying distribution is exactly Pareto. Given its common use in macroeconomics, we interchangeably refer to the Pareto exponent and the tail index in this paper when there is no confusion. } 
of the wealth distribution of the richest is significantly associated with the maximum marginal income tax rate. 
%significant effects (\edits{correlation}?) of the maximum marginal income tax rate on the tail index or the Pareto exponent of the wealth distribution of the richest.
Moreover, we further analyze counterfactual levels of economic (in)equality that would result under counterfactual policies that set various levels of the maximum marginal income tax rate.
Under low levels of the top tax rate, such as 30--40\% as is the case with the United States today, we find that the density of wealthy population entails Pareto exponents of 1.5--1.8 regardless of the current macroeconomic states.
On the other hand, under high levels of the top tax rate, such as 60\% as was the case with the United States before 1980, we find that the density of wealthy population would entail Pareto exponents of around 2.1 regardless of the current macroeconomic states. 
%In other words, the maximum marginal income tax rate needs to be raised to this level to ensure a finite second moment of the wealth distribution.
%(\edits{I suggest deleting this sentence--it's unclear what ``optimal'' means and we are just doing econometrics}) \citet[p.~660]{piketty2013capital} argues that ``the optimal top tax rate in the developed countries is probably above 80 percent.''
%Extrapolating our counterfactual analysis to this level of the maximum marginal income tax rate, we predict that our economy would achieve the Pareto exponent of 2.6 or above under the optimal policy suggested by \citeauthor{piketty2013capital}.

To explain these empirical findings, we present a simple dynamic general equilibrium model that features entrepreneurs that are subject to capital income risk and workers. We derive the equilibrium wealth Pareto exponent in closed-form and show that it depends only on idiosyncratic volatility, bankruptcy rate, and the tax rate, and that it is increasing in the tax rate, consistent with our empirical findings. We also derive the welfare of workers and entrepreneurs as a function of the tax rate and show that the worker welfare is decreasing in the tax rate (because high tax leads to low capital stock and wage), whereas the entrepreneur welfare is single-peaked (because high tax leads to low capital stock but provides insurance against capital income risk through loss deductions).

\paragraph{Literature}
Estimation and inference about tail index have spawned extensive literature in econometrics and statistics. 
We provide a limited discussion here. 
First, consider estimating the tail index of some underlying distribution $F_Y$ with a cross-sectional sample $\set{Y_1,Y_2,\dots,Y_n}$. 
The widely used estimator due to \citet{Hill1975} essentially fits the largest $k$ order statistics to the Pareto distribution and conducts the MLE. 
The asymptotic properties have been extensively explored by, for example, \citet{Hall1982}, \citet{Mason1982}, \citet{DavisResnick1984}, \citet{HaeuslerTeugels1985}, and \citet{HallWelsh1985}. 
Other estimators include, among many others, \citet{Pickands1975}, \citet{Smith1987MLE}, \citet{Gabaix2011rank}, and \citet{NicolauRodrigues2019}. 
See also \citet[Ch.~3]{deHaan2006book} for an overview. 
Second, Hill's estimator has been extended to numerous dependent and heterogeneous processes. 
Examples include \citet{Rootzen1990}, \citet{Hsing1991,Hsing1993}, and \citet{Resnick1995,Resnick1998}. 
More recently, \citet{Hill2010} introduces and studies the extremal versions of the classic mixingale and near epoch dependent conditions. 
These conditions cover a broad range of stochastic processes, including any geometrically mixing process. 
\citet{Hill2015} further studies filtered dependent data sets, which cover regression residuals, GARCH filters, and weighted sums based on an optimization problem such as optimal portfolio selection. 
Third, Hill's estimator has been extended to scenarios with data truncation. % from the right/top. 
The threshold for truncation can be a fixed constant or a latent random variable -- see, for example, \citet{Beirlant2001}, \citet{Aban2006}, \citet{Ndao2014}, \citet{Benchaira2016kernel,Benchaira2016tail}, \citet{Worms2016}, and \citet{WangXiao2021}. 
Unlike all these papers, our framework involves truncation based on the order $k$ (for order statistics) rather than a constant or random threshold.

Now, consider estimating the tail index of conditional distribution $F_{Y|X}$ given $X = x$ for some query point $x$. 
This problem is substantially more difficult than its unconditional counterpart since the tail index is in principle an unknown function of $X$. 
Given an independent sampling of $(Y,X)$, \citet{WangTsai2009} and \citet{WangLi2013} assume that the tail index conditional on $X=x$ takes the single index form $\exp(x^\intercal \theta_0)$ for some pseudo-true parameter $\theta_0$ and develop regression-type estimators for $\theta_0$. 
\citet{Gardes2008} and \citet{Gardes2012} develop nonparametric kernel-based estimators for the conditional tail index by using the observations with $X_i$ within a local neighborhood of the query point $x$. 
Recently, \citet{LiLengYou2022} develop a semiparametric estimator of the tail index, which unifies these two types of estimators. 

In summary, the existing estimation methods in the above-mentioned papers all rely on the so-called increasing-$k$ asymptotics, where $k \to \infty$ and $k/n \to 0$ as $n \rightarrow \infty$. 
These two conditions on $k$ indicate a delicate balance in controlling the bias and the variance. 
In practice, $k$ is considered as a tuning parameter, whose optimal choice is of the order $O(n^{-\kappa})$ for some $\kappa \in (0,1/2)$ \citep[e.g.,][]{Hall1982}. 
In contrast, our context involves a fixed $k$ at 400 while $n$ is gigantic (over a hundred million). 
This key difference motivates us to adopt the fixed-$k$ asymptotic framework, which is proposed by \citet{MuellerWang2017} and further studied by \citet{sasaki2020testing,SasakiWang2022}. 
In a broad sense, such asymptotic framework with a fixed tuning parameter has been explored in other fields of econometrics \citep[e.g.,][]{KieferVogelsang2005, CattaneoCrumpJansson2014, Bollerslev2021}.

Consistent estimation of the tail index is out of the question when $k$ is fixed, and hence all these papers on the fixed-$k$ approaches \citep{MuellerWang2017,sasaki2020testing,SasakiWang2022} focus on asymptotic inference without estimation. 
We sidestep this issue by exploiting the time-series variation in $t=1,\dots,T$ with $T \rightarrow \infty$ so that the tail index can still be consistently estimated.  
Of course, our method requires additional data set beyond a cross-sectional sample of $(Y,X)$. 
Specifically, we observe extremal order statistics $\set{Y_{(i)t}}_{i=1}^k$ for each time $t$, and $X_t$ is a time series that is common across all $i$ for each $t$. 
Our proposed method is motivated by this peculiar sampling process, but is not limited to it. 
We discuss potential extensions to study firm size, birthweight, medical insurance claim, and earnings risk in the concluding remarks. 

Finally, also related is the literature about tail features with panel or repeated cross-sectional data sets, which contains relatively few econometric methods. 
\citet{HillShneyerov2013} propose a tail index estimator for an unbalanced panel for first-price auction data. 
\citet{KellyJiang2014} apply Hill's estimator to study the tail heaviness of cross-sectional stock returns in each month. 
\citet{SasakiWang2022} propose confidence intervals for the conditional extreme quantiles of $Y_{it}$ given $X_{it}$ in panel or repeated cross-sectional data sets. 
\citet{Sarpietro2022} extend Hill's estimator and study conditional earnings risks. 
None of these papers consider modeling the tail index as a function of covariates and hence their proposed methods cannot be directly applied to our context.   

%Note that the existing nonparametric methods \citep[e.g.,][]{Martins-Filho2015nonpara, Martins-Filho2018nonpara} model tail index to be constant in controls.
%typically admit a location and scale representation, which implicitly implies a constant tail index. 
% Our MLE complements these methods by allowing the tail index, as well as the location and scale, to vary with $X_{t}$. 

\paragraph{Organization}
We present the econometric method and theory in Section \ref{sec:fixedk}.
This is followed by a data description and an empirical analysis in Sections \ref{sec:data} and \ref{sec:results}, respectively. 
Section \ref{sec:model} contains a macroeconomic model, and Section \ref{sec:summary} concludes.
%The mathematical proofs are found in the appendix.
%Additional theoretical and estimation results as well as simulations are presented in the supplementary material.
Mathematical proofs, additional theoretical and estimation results as well as simulations are presented in the Appendix.

%%%%%%%%%%%%%%%%%%%%%%%%%%%%%%%%%%%%%%%%%%%%%%%%%%%%%%%%%%%%%%%%%%%%%%
\section{The fixed-$k$ tail regression}\label{sec:fixedk}
%%%%%%%%%%%%%%%%%%%%%%%%%%%%%%%%%%%%%%%%%%%%%%%%%%%%%%%%%%%%%%%%%%%%%%

In this section, we propose a tail regression to model the dependence of the tail index (the reciprocal of the Pareto exponent) on controls, and develop a method for estimating its model parameters along with their standard errors.
Consequently, we may conduct inference on the partial effect of the top marginal tax rate on the tail index and inference on counterfactual Pareto exponents under alternative levels of the top marginal tax rate that are considered by policy makers.

\subsection{Basic setup}

Consider a repeated cross section $\set{ Y_{it} \in \R : i=1,\dots ,n, t=1,\dots ,T}$ of the values $Y_{it}$ of wealth, where $i$ indexes individuals and $t$ indexes calendar years.
For our data, the year 1982 is set as $t=1$ and there are $T=36$ periods until 2017. 
This repeated cross section may be considered to consist of the universe of all the individuals in the population, but our econometric method to be presented below requires us to observe only the top $k$ values of wealth in each year $t$ for some predetermined natural number $k$.
For our sample of Forbes 400, we can set $k$ to be any natural number that is smaller than or equal to $400$ (but at least 3 for methodological reasons).
As such, we only need to observe the top $k$ observations of $Y_{it}$, and do not need to observe $Y_{it}$ for the rest of the population.\footnote{For now we suppose that $Y_{it}$ are observed without errors. However, as we show in Appendix \ref{sec:robust}, our estimation and inference method is valid under measurement errors, which are likely to arise in the Forbes 400 data set.}

In addition to the repeated cross section of the values of wealth, we also observe a time series $\set{X_t \in \R^p : t = 1, \dots, T}$ of a vector consisting of policy factors and macroeconomic indices such that $X_t$ may affect the right tail of the cross sectional distribution of $\set{Y_{it} \in \R : i = 1,\dots,n}$.
For our empirical question, $X_t$ includes the averages of $\Delta t$-year lags of the maximum marginal income tax rate, idiosyncratic volatility, and bankruptcy rate for various values, where $\Delta t \in \set{3,5,7,9}$.\footnote{This choice of controls is motivated by the economic model to be presented in Section \ref{sec:model}.}
In this setup, we are interested in the right tail behavior of the conditional distribution $F_{Y_{it}|X_t}$ of $Y_{it}$ given $X_t$.
Specifically, when $F_{Y_{it}|X_t=x}$ is approximately Pareto in the right tail, the Pareto exponent should be considered as a function of $x$, that is, $\alpha = \alpha(x)$.  %we are interested in the tail index $\xi(x)$ of the conditional distribution $F_{Y_{it}|X_t=x}$ of $Y_{it}$ given $X_t=x$.
This feature tells us how the policy and macroeconomic predictors $X_t$ may (or may not) explain the heaviness of the right tail of the wealth distribution.

For notational simplicity, we introduce the reciprocal of the Pareto exponent $\xi(x) \coloneqq 1/\alpha(x)$, where $\xi(x)$ is referred to as the tail index (conditional on $X_t=x$).
To model the effect of $X_t$ on $\xi(\cdot)$, we assume that the tail index takes the following single-index form 
\begin{equation}\label{eq:tail_regression}
\xi(x)=\Lambda (x^\intercal \theta_0) \quad \text{for some true parameter $\theta_0\in \Theta \subset \R^p$,}
\end{equation}
where $\Lambda (\cdot )$ is some link function chosen by the econometrician.
We use the standard logistic link such that $\Lambda (s) = 1/(1+\exp(-s))$ in our analysis to guarantee that $\xi(x) \in (0,1)$ for all $x$. 
Other choices such as the Gaussian link can also be used and accordingly lead different pseudo-true values of $\theta_0$.     
By construction, we have $\alpha(X_t)=1/\xi(X_t) > 1$, which is consistent with the theoretical formula \eqref{eq:PE} to be presented in Section \ref{sec:model}, as well as our empirical findings.
Within this framework, our problem of learning about the tail index $\xi(x)$ as a function of $x$ boils down to learning about the parameter vector $\theta_0$.
To this goal, we propose a method to estimate $\theta_0$.
With an estimate $\hat\theta$ of $\theta_0$ along with its standard errors, we can conduct statistical inference about which of the policy and macroeconomic factors $X_t$ explain the heaviness of the right tail of the wealth distribution.
Furthermore, we can also make predictions about the right tail behavior $\hat\xi(x) = \Lambda(x^\intercal \hat\theta_0)$ of the wealth distribution under counterfactual values of $x$, such as counterfactual policies of maximum marginal income tax rates.

\subsection{Econometric method}\label{subsec:method}

Not surprisingly, a na\"ive regression analysis will not estimate $\theta_0$ for this non-standard formulation of the tail regression.
We therefore propose the following econometric method of estimation and inference about $\theta_0$. 
Let
\begin{align}
&f_{\bV^*|\Lambda (X_t^\intercal \theta)}(1,v_{2}^*,\dots, v_{k-1}^*,0) \notag \\
&=\Gamma( k) \int_0^{\infty }s^{k-2}\exp \left( -\left(1+\frac{1}{\Lambda (X_t^\intercal \theta)}\right)\sum_{j=1}^{k}\log \left(1+\Lambda (X_t^\intercal \theta)v_j^*s\right)\right) \diff s,
\label{fvstar}
\end{align}
where $\Gamma(\cdot)$ denotes the Gamma function.\footnote{This density can be numerically computed by standard methods such as Gaussian quadrature and the trapezoidal rule.} 
Let 
\begin{equation*}
\mathcal{I}(\theta_0) =\left. \mathbb{E}\left[ \frac{\partial
\log f_{\mathbf{V}^*|\Lambda \left( X_{t}^{\intercal }\theta \right)
}\left( \mathbf{V}^*\right) }{\partial \theta }\frac{\partial \log f_{%
\mathbf{V}^*|\Lambda \left( X_{t}^{\intercal }\theta \right) }\left( 
\mathbf{V}^*\right) }{\partial \theta }\right] \right\vert _{\theta
=\theta _{0}}
\end{equation*}%
denote the Fisher information matrix. 
We also write $s_{t}\left( \theta \right)=\log f_{\mathbf{V}^*|\Lambda \left( X_{t}^{\intercal }\theta \right)
}\left( \mathbf{Y}_{t}^*\right) $ and $\dot{s}_{t}\left( \theta
\right) =\partial s_{t}\left( \theta \right) /\partial \theta $.
Finally, let 
\begin{equation*}
\mathcal{W}(\theta_0) =\lim_{T\rightarrow \infty }\Var\left[
T^{1/2}\sum_{t=1}^{T}\dot{s}_{t}(\theta_0) \right]
\end{equation*}%
denote the long-run variance-covariance matrix.

\begin{algorithm}[Fixed-$k$ Tail Regression]\label{alg:regression}
Let $Y_{it}$ denote the level of wealth of individual $i$ in year $t$.
Let $X_t$ denote the vector of controls in year $t$.
\begin{enumerate}[1.]
	\item Sort $\set{Y_{1t},\dots ,Y_{nt}}$ in descending order as $Y_{(1)t}\ge Y_{(2)t}\ge \dots \ge Y_{(n)t}$ for each year $t$ and take the largest $k\ge 3$ order statistics $\bY_t=\left( Y_{(1)t},Y_{(2)t},\dots ,Y_{(k)t}\right)^\intercal$.
	\item Form the self-normalized statistics
\begin{equation}
\bY_t^* =\left(1, \frac{Y_{(2)t}-Y_{(k)t}}{Y_{(1)t}-Y_{(k)t}},\dots ,\frac{Y_{(k-1)t}-Y_{(k)t}}{Y_{(1)t}-Y_{(k)t}},0\right).\label{eq:self_normalized}
\end{equation}
	\item Estimate $\theta_0$ by the maximum likelihood estimator (MLE) $\hat\theta = \arg\min_{\theta \in \Theta} S_T(\theta)$, where
\begin{equation}\label{eq:S}
S_{T}(\theta)=-\frac{1}{T}\sum_{t=1}^{T}\log f_{\bV^*|\Lambda (X_t^\intercal \theta )}(\bY_t^*).
\end{equation}
	\item  Estimate the standard error of $\hat{\theta}$ by $\diag\left( T^{-1} \hat{\Sigma} (\theta_0)\right)^{1/2}$, where
\begin{equation*}
\hat{\Sigma} (\theta_0) = \mathcal{\hat{I}}(\theta _{0}) ^{-1}\mathcal{\hat{W}}(\theta_0) \mathcal{\hat{I}}(\theta_0) ^{-1}, 
\end{equation*}	
\begin{equation*}
\mathcal{\hat{I}}(\theta_0) =\frac{1}{T}\sum_{t=1}^{T}\left. 
\frac{\partial \log f_{\mathbf{V}^*|\Lambda( X_{t}^{\intercal
}\theta) }(\mathbf{Y}_{t}^*) }{\partial \theta }%
\frac{\partial \log f_{\mathbf{V}^*|\Lambda( X_{t}^{\intercal}\theta) }(\mathbf{Y}_{t}^*) }{\partial \theta }%
\right\vert _{\theta =\hat{\theta}},
\end{equation*}
and $\mathcal{\hat{W}}(\theta_0) $ denotes some heteroskedasticity and autocorrelation consistent (HAC) estimator of the long-run variance -- see a concrete suggestion below.
\end{enumerate}
\end{algorithm}

\noindent
Discussions of the density \eqref{fvstar} and formal asymptotic statistical theories that guarantee that this proposed method works are presented in Section \ref{subsec:theory} -- see Theorem \ref{thm:asym} in particular. 
In our empirical analysis, we use the popular estimator proposed by \citet{NeweyWest1987} for $\mathcal{\hat{W}}(\theta_0) $. 
Let 
\begin{equation*}
\hat{\Omega}_{j}=\frac{1}{T-p}\sum_{t=j+1}^{T}\dot{s}_{t}(\hat{\theta}) \dot{s}_{t-j}(\hat{\theta}) ^{\intercal }
\end{equation*}%
for $j=0,1,2,\dots,T-1$. Then, the Newey-West estimator is constructed as%
\begin{equation*}
\mathcal{\hat{W}}(\theta_0) =\hat{\Omega}_{0}+%
\sum_{j=1}^{L}w\left( j,L\right) \left[ \hat{\Omega}_{j}+\hat{\Omega}%
_{j}^{\intercal }\right] ,
\end{equation*}%
where $w\left( j,L\right) =1-\frac{j}{L+1}$. 
The truncation parameter $L$ is chosen as $3 \approx 0.75T^{1/3}$ following \citet{Andrews1991}.

Once the estimate $\hat\theta$ of the parameter vector $\theta_0$ is obtained, we can in turn estimate the marginal effect
$
\partial \xi(X_t) / \partial X_{jt}
=
\theta_{0j} \Lambda'(X_t^\intercal\theta_0)
$
of the $j$-th predictor $X_{jt}$ (e.g., the maximum marginal tax rate) on the tail index $\xi(X_t^\intercal\theta_0)$ in year $t$ by the plug-in counterpart
\begin{equation}\label{eq:marginal}
\widehat{\partial \xi(X_t) / \partial X_{jt}}
=
\hat\theta_j \Lambda'(X_t^\intercal\hat\theta).
\end{equation}
Its standard error can be obtained by the delta method provided that the link function $\Lambda$ is twice continuously differentiable in a neighborhood of $X_t^\intercal\theta_0$ -- see Corollary \ref{corollary:marginal} in Section \ref{subsec:theory}.
Similarly, letting $X_t^\mathrm{cf}$ denote a counterfactual of $X_t$ (e.g., by replacing the actual maximum marginal tax rate by a counterfactual rate under consideration by the policy maker), we can estimate the counterfactual value of the tail index $\xi(X_t^\mathrm{cf})$ by
\begin{equation}\label{eq:counterfactual}
\widehat{\xi(X_t^\mathrm{cf})}
=
\Lambda((X_t^\mathrm{cf})^\intercal\hat\theta).
\end{equation}
Its standard error can be obtained, again, by the delta method provided that the link function $\Lambda$ is continuously differentiable in a neighborhood of $(X_t^\mathrm{cf})^\intercal\theta_0$ -- see Corollary \ref{corollary:level} in Section \ref{subsec:theory}.
When one is interested in the Pareto exponent as opposed to the tail index, its counterfactual value $\alpha(X_t^\mathrm{cf}) = 1/\xi(X_t^\mathrm{cf})$ can be estimated by
\begin{equation}\label{eq:counterfactual_pareto}
\widehat{\alpha(X_t^\mathrm{cf})}
=
1/\Lambda((X_t^\mathrm{cf})^\intercal\hat\theta).
\end{equation}
See Corollary \ref{corollary:pareto} in Section \ref{subsec:theory} for its standard error.

%%%%%%%%%%%%%%%%%%%%%%%%%%%%%%%%%%%%%%%%%%%%%%%%%%%%%%%%%%%%%%%%%%%%%%
%\section{Econometric theory of the tail regression}\label{subsec:theory}
%%%%%%%%%%%%%%%%%%%%%%%%%%%%%%%%%%%%%%%%%%%%%%%%%%%%%%%%%%%%%%%%%%%%%%

\subsection{Econometric theory}\label{subsec:theory}

In this section, we present the asymptotic theory to guarantee that the method proposed in Section \ref{subsec:method} works in large sample.
We adopt the standard definition of $\alpha$-mixing as follows.
\begin{definition}
	A sequence $(X_{1}, X_2, X_3,\dotsc)$ is said to be $\alpha$-mixing if
	 \begin{equation*}
		\tilde{\alpha}(h)=\sup_{k\ge1}\sup_{\mathcal{A}\in \mathcal{F}_{1}^{k}, \mathcal{B}\in \mathcal{F}_{k+h}^{\infty}}\abs{\mathcal{P}(\mathcal{A}\cap \mathcal{B})-\mathcal{P}(\mathcal{A})\mathcal{P}(\mathcal{B}) } \to 0, \ as \ h \to \infty
	\end{equation*}
where $\mathcal{F}_{i}^{j}$ denotes the $\sigma$-field generated by $(X_{i}, X_{i+1},\dots, X_{j})$. 
\end{definition}
Letting $\alpha (x)=1/\xi (x)$, we assume that our data set conforms with the following four sets of conditions.
\begin{condition}\label{cond1}
 $X_{t}$ is strictly stationary and $\alpha $-mixing with the mixing coefficients satisfying that $\sum_{h=1}^{\infty }h^2 \tilde{\alpha} \left(h\right) ^{1/3}<\infty $. 
 Conditional on $X_{t}=x$, $\{Y_{it}\}_{i=1}^{n}$ are \iid draws from some distribution $F_{Y|X=x}(y)$. 
\end{condition}

\begin{condition}\label{cond2}
$F_{Y|X=x}(y)$ satisfies 
\begin{equation*}
1-F_{Y|X=x}(y)=c(x)y^{-\alpha (x)}(1+d(x)(y)^{-\beta(x)}+r(x,y)),
\end{equation*}
where $c(\cdot)>0$ and $d(\cdot)$ are uniformly bounded
between $0$ and $\infty$ and continuously differentiable with uniformly
bounded derivatives, $\alpha (\cdot)>1$ and $\beta (\cdot)>0$ are continuously differentiable functions, and $r(x,y)$
is continuously differentiable with bounded derivatives with respect to both $x$ and 
$y$ with
\begin{equation*}
\limsup_{y\to \infty}\sup_{x}\abs{r(x,y)y^{\beta
(x)}} \to 0.
\end{equation*} 

\end{condition}

\begin{condition}\label{cond3}
$\theta_0$ is in the interior of $\Theta$, a compact subset of $\R^p$ and $X_{t} \in \mathcal{X}$, a compact subset of $\R^p$. Moreover, $\Lambda$ is continuously differentiable with uniformly bounded first derivative and satisfies that $0<\inf_{(x,\theta) \in \mathcal{X}\times\Theta}\Lambda(x^{\intercal}\theta) \le \sup_{(x,\theta) \in \mathcal{X}\times\Theta} \Lambda(x^{\intercal}\theta)<1$, and the Fisher information matrix 
\begin{equation*}
\mathcal{I}(\theta_0) =\E\left[ \frac{\partial \log  f_{\bV^*|\Lambda(X_{t}^{\intercal }\theta_0) }(\bV^*) }{\partial \theta }\frac{\partial \log f_{\bV^*|\Lambda (X_{t}^{\intercal }\theta _{0}) }( \bV^*) }{\partial \theta ^{\intercal }}\right] 
\end{equation*}
is positive definite. 
\end{condition}

\begin{condition}\label{cond4}
$n \to \infty$ and $T\sup_{x}n^{-\beta(x) /\alpha (x) }\to 0$ as $T\to \infty$.
\end{condition}

%We provide some discussions about these conditions. 
Condition \ref{cond1} assumes that the vector of macroeconomic indicators $X_t=x$ is a stationary and mixing time series, and that individual wealth is randomly drawn from the population wealth distribution $F_{Y|X=x}(\cdot)$ conditionally on $X_t=x$. 
The $\alpha $-mixing condition can be replaced with many other weak dependence structures \citep[cf.][]{Bradley2007}. 
Condition \ref{cond2} assumes that such conditional wealth distribution has an \textit{approximate} Pareto tail with the approximation error characterized by the second-order parameter $\beta(x)$ and the remainder function $r(x,y)$.
The unconditional version of this condition has been commonly used in the statistics literature to study tail features: see, for example, \citet{Hall1982}.
This condition is mild and satisfied by many commonly used distributions, including, for example, joint Student-t distributions. 
See the discussion in \citet{SasakiWang2022}.  
The condition that $\alpha (\cdot)>1$ is due to our choice of the logistic link function and is coherent with our empirical findings since the mean of the wealth distribution is found to be finite. 
We can relax it to $\alpha (\cdot)>0$ by choosing other link functions such as $\Lambda(s) = \exp(s)$.  
Condition \ref{cond3} requires that the true parameter is not on the boundary of the parameter space and the covariates have a bounded support. 
Also, it requires that the space of the tail index $\xi$ is strictly within $(0,1)$ and the Fisher information matrix is invertible. 
These conditions are used to identify the pseudo-true parameter $\theta_0$, which maximizes the limit of the criteria function $S_T (\theta)$. 
%The compact support assumption is sufficient but not necessary.
%It is imposed to simplify the proof and plausibly satisfied by our data set, especially in our sampling periods. 
%Note that Conditions \ref{cond2} and \ref{cond3} are mild and satisfied by many commonly used distributions, including, for example, joint normal and joint Student-t distributions. 
%See the discussion in \citet{SasakiWang2022}. 
Condition \ref{cond4} assumes a large $n$ and large $T$ asymptotic framework and requires that the magnitude of $T$ is not too large relative to $n$. 
For technical simplicty, we consider $n=n_T$ as a function of $T$ so that the process $(\bY^{*}_t, X_t)$ is a mixing triangular array. 
%Then the standard law of large numbers and central limit theorys can be applied. 
Since $n$ is gigantic in practice, we find through Monte Carlo studies that $T=36$ as in our Forbes 400 data set is sufficient to guarantee the good performance of our estimator.
See Appendix \ref{sec:simulation} ahead for more details.   

Recall that in Algorithm \ref{alg:regression} we sort $\set{Y_{1t},\dots ,Y_{nt}}$ in descending order as $Y_{(1)t}\ge Y_{(2)t}\ge \dots \ge Y_{(n)t}$ and take the largest $k$ order statistics 
$\bY_t=( Y_{(1)t},Y_{(2)t},\dots ,Y_{(k)t})
^\intercal.$
By Conditions \ref{cond1} and \ref{cond2}, $\{Y_{it}\}_{i=1}^n$ are \iid across $i$ conditional on $X_t$ at each $t$. 
Then extreme value theory \citep[e.g.,][Theorem~2.1.1]{deHaan2006book} implies that there exist sequences of constants $a_n$ and $b_n$ such that conditional on $X_t$
for each $t$ and any \emph{fixed} $k,$
\begin{equation}\label{eq:limit_dist}
\frac{\bY_t-b_n}{a_n}\Rightarrow \bV_t
\end{equation}
as $n\to \infty$, where $\bV_t=\left( V_{1t},V_{2t},\dots ,V_{kt}\right) ^{\intercal }$ is jointly EV distributed with PDF 
\begin{equation}
f_{\bV|\xi }(v_{1},\dots ,v_k)=G_{\xi}(v_k)\prod_{j=1}^{k}g_{\xi }(v_j)/G_{\xi }(v_j) \label{fv}
\end{equation}
with $g_{\xi }(v)=\partial G_{\xi }(v)/\partial v$ and
\begin{equation*}
G_{\xi }(v)=\begin{cases*}
\exp (-(1+\xi v)^{-1/\xi }) & if $\xi \neq 0$ and $1+\xi v\ge 0$,\\
\exp (-\e^{-v}) & if $\xi = 0$ and $v\in \R$.
\end{cases*}
\end{equation*}
%The vector $\bV_t$ is strictly stationary and weakly dependent as implied by Condition \ref{cond1} \textcolor{red}{(we may need to show this)}. 
The tail index parameter $\xi$ characterizes the tail shape of $Y_{it}$, and it is a function $\xi(X_t)=\Lambda (X_t^\intercal \theta_0)$  of $X_t$.

The constants $a_n$ and $b_n$ depend on the unknown distribution of $Y_t$ and are difficult to estimate. This is why we consider the self-normalized statistics \eqref{eq:self_normalized} to eliminate them. By the variable transformation, we obtain
\begin{align*}
\bY_t^*|X_t =& \left. \left( 1,\frac{Y_{(2)t}-Y_{(k)t}}{Y_{(1)t}-Y_{(k)t}},\dots ,\frac{Y_{(k-1)t}-Y_{(k)t}}{Y_{(1)t}-Y_{(k)t}},0\right) \middle| X_t \right.\\
\Rightarrow \bV_t^* 
\coloneqq &\left( 1,\frac{V_{2t}-V_{kt}}{V_{1t}-V_{kt}},\dots ,\frac{V_{k-1,t}-V_{kt}}{V_{1t}-V_{kt}},0\right),
\end{align*}
where the PDF $f_{\bV_t^*|\xi(X_t)}$ of the non-degenerate coordinates of $\bV^*$ can be calculated as in \eqref{fvstar}.

Define the negative log-transformed likelihood function as in \eqref{eq:S}. 
Minimizing $S_{T}(\theta) $ over $\theta $ results in the
approximate MLE $\hat{\theta}=\arg \min_{\theta \in \Theta }S_{T}(\theta) $. The following theorem establishes the asymptotic normality
of $\hat{\theta}$.

\begin{theorem}\label{thm:asym}
Suppose that Conditions \ref{cond1}-\ref{cond4} hold. Then as $T\to \infty$ and $n\to \infty$, for any fixed $k$, 
\begin{equation*}
T^{1/2}(\hat{\theta}-\theta_0) \Rightarrow \cN\left(
0,\Sigma (\theta_0) \right) ,
\end{equation*}%
where 
\begin{equation*}
\Sigma \left( \theta_0\right) = \cI(\theta_0)^{-1}\mathcal{W}(\theta_0)\cI(\theta_0)^{-1}. 
\end{equation*}%

\end{theorem}

Theorem \ref{thm:asym} shows that the econometric method proposed in Section \ref{subsec:method} is guaranteed to work in large sample. Applying this theorem with the delta method to the marginal effect \eqref{eq:marginal}, we obtain the following corollary.

\begin{corollary}\label{corollary:marginal}
Suppose that Conditions \ref{cond1}-\ref{cond4} hold. 
If $\Lambda$ is twice continuously differentiable, then as $T\to
\infty$ and $n\to \infty$, for any fixed $k$,
\begin{equation*}
T^{1/2}\left(\widehat{\partial \xi(x) / \partial x_j} - \partial \xi(x) / \partial x_j\right)
\Rightarrow \cN\left(0,\Xi_j(x)^\intercal \Sigma(\theta_0) \Xi_j(x) \right) ,
\end{equation*}
where
%$
%\Xi_j(x) 
%= 
%\left(x_j \theta_{0j} \Lambda''(x^\intercal\theta_0) + \Lambda'(x^\intercal\theta_0)\right) e_j
%$
$
\Xi_j(x) 
= 
x \theta_{0j} \Lambda''(x^\intercal\theta_0) + e_j \Lambda'(x^\intercal\theta_0)
$
with $e_j$ denoting the $j$-th standard unit vector.
\end{corollary}

\noindent
Similarly, applying the theorem with the delta method to the counterfactual level \eqref{eq:counterfactual}, we obtain the following corollary.

\begin{corollary}\label{corollary:level}
Suppose that Conditions \ref{cond1}-\ref{cond4} hold. 
As $T\to
\infty$ and $n\to \infty$, for any fixed $k$,
\begin{align*}
T^{1/2}\left(\widehat{\xi(x)} - {\xi(x)}\right)
\Rightarrow \cN\left(0,\Xi(x)^\intercal \Sigma(\theta_0) \Xi(x) \right) ,
\end{align*}
where
$\Xi(x) 
= x \Lambda'(x^\intercal\theta_0)$.
\end{corollary}

\noindent
Likewise, for the counterfactual value \eqref{eq:counterfactual_pareto} of the Pareto exponent, we obtain the following corollary.

\begin{corollary}\label{corollary:pareto}
Suppose that Conditions \ref{cond1}-\ref{cond4} hold. 
As $T\to
\infty$ and $n\to \infty$, for any fixed $k$,
\begin{equation*}
T^{1/2}\left(\hat{\alpha(x)} - {\alpha(x)}\right)
\Rightarrow \cN\left(0,A(x)^\intercal \Sigma(\theta_0) A(x) \right) ,
\end{equation*}
where
$
A(x) 
= 
x \Lambda'(x^\intercal\theta_0) / \Lambda(x^\intercal\theta_0)^2.
$
\end{corollary}

%%%%%%%%%%%%%%%%%%%%%%%%%%%%%%%%%%%%%%%%%%%%%%%%%%%%%%%%%%%%%%%%%%%%%%
\section{Data and preliminary analysis}\label{sec:data}
%%%%%%%%%%%%%%%%%%%%%%%%%%%%%%%%%%%%%%%%%%%%%%%%%%%%%%%%%%%%%%%%%%%%%%

%%%%%%%%%%%%%%%%%%%%%%%%%%%%%%%%%%%%%%%%%%%%%%%%%%%%%%%%%%%%%%%%%%%%%%%
%\subsection{Wealth and tax data}
%%%%%%%%%%%%%%%%%%%%%%%%%%%%%%%%%%%%%%%%%%%%%%%%%%%%%%%%%%%%%%%%%%%%%%%

In this section, we introduce the data set and provide some preliminary analysis.
First, Forbes 400 provides the list of individuals ranked in the top 400 in the United States in terms of the wealth, along with the value of the wealth for each individual in the list.
These values are calculated by the staff in Forbes Magazine who ``[\ldots] pored over hundreds of Securities Exchange Commission documents, court records, probate records, federal financial disclosures and Web and print stories.''\footnote{We thank Matthieu Gomez for sharing this data set.}
We use the same sample of Forbes 400 as the one that has also been used in the literature mentioned in the literature review in Section \ref{sec:introduction}. 
This sample consists of repeated cross sections of 400 individuals for 36 years from 1982 to 2017.
We use the raw values provided in data set without inflating or deflating them, because this data set will be used only for the purpose of estimating tail indices and Pareto exponents, which are invariant from scales.
Table \ref{tab:forbes400} presents a partial list of the raw values in millions of U.S. dollars.

\begin{table}[!htb]
	\centering
	\renewcommand{\arraystretch}{1}
	\resizebox{\linewidth}{!}{%
		\begin{tabular}{lccccccccccccc}
		\toprule
		&& \multicolumn{12}{c}{Rank}\\
		\cmidrule{3-14}
			Year && 1    & 2    & 3    & 4    & 5    & 10   & 20   & 50   & 100 & 200  & 300  & 400\\
		\midrule
			1985 && 2800 & 1800 & 1500 & 1400 & 1300 & 1000 & 875  & 600  & 360  & 233  & 183  & 150\\
			1990 && 5600 & 3343 & 2870 & 2650 & 2600 & 2500 & 2000 & 1300 & 730 & 450 & 340 & 260\\
			1995 && 14800 & 11800 & 6700 & 6100 & 4800 & 4300 & 3000 & 1800 & 900 & 600 & 435 & 340\\
			2000 && 63000 & 58000 & 36000 & 28000 & 26000 & 17000 & 10000 & 4700 & 2600 & 1500 & 980 & 725\\
			2005 && 51000 & 40000 & 22500 & 18000 & 17000 & 15400 & 10000 & 4200 & 2500 & 1600 & 1200 & 900\\
			2010 && 54000 & 45000 & 27000 & 24000 & 21500 & 18000 & 12400 & 5300 & 3200 & 2000 & 1400 & 1000\\
			2015 && 76000 & 62000 & 47500 & 47000 & 41000 & 33300 & 23400 & 9000 & 5000 & 3300 & 2300 & 1700\\
		\bottomrule
		\end{tabular}
	}
	\caption{The values of the wealth in millions of U.S. dollars at the ranks 1, 2, 3, 4, 5, 10, 20, 50, 100, 200, 300 and 400 in Forbes 400 in 1985, 1990, 1995, 2000, 2005, 2010, and 2015.}${}$
	\label{tab:forbes400}
\end{table}

Second, we use the data about the top marginal income tax rates available from Tax Policy Center, Urban Institute \& Brookings Institution.\footnote{\url{https://www.taxpolicycenter.org/statistics/historical-highest-marginal-income-tax-rates}}
This tax data is available from 1913 to 2020, thus sufficiently covering the period 1982--2017 of our Forbes 400 sample.
Figure \ref{fig:series_tax} shows the time series of the top marginal income tax rate in the United States.
The rates were high at 70\% or above until 1980, followed by a rapid decline during the 1980s down to 28\% in 1990.
%As ``the magnitude and character of the tax change was largely unexpected before 1986'' \citep{Feldstein1995} it provides an ``important natural experiment'' \citep{Feldstein1995} to study effects of tax rates on wealth inequality.
The rates have eventually fluctuated within a narrow range of 30--40\% after 1990.

\begin{figure}[!htb]
\centering
\includegraphics[width=0.5\linewidth]{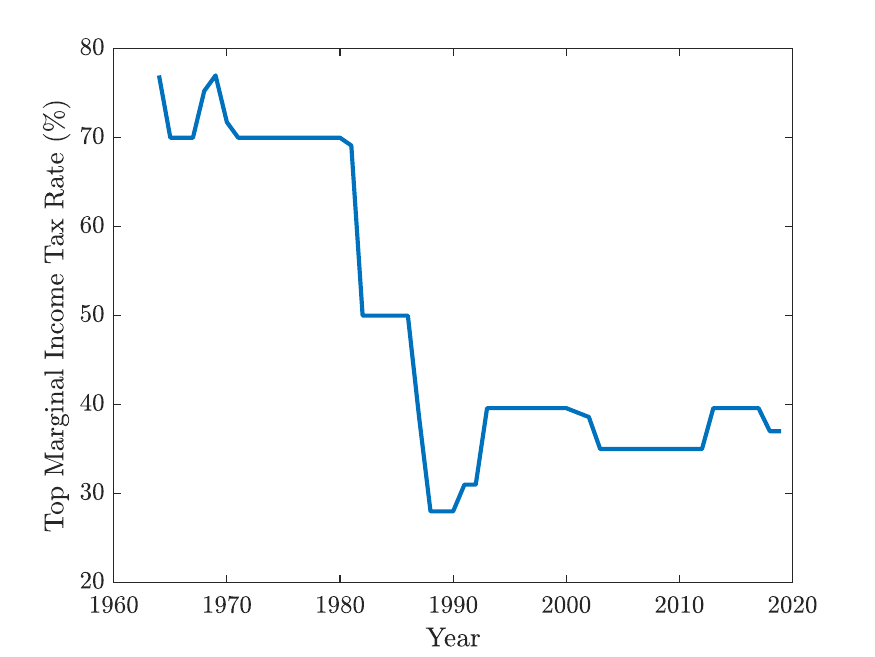}
\caption{Time series of the top marginal income tax rate in the United States.}
\label{fig:series_tax}
\end{figure}

%%%%%%%%%%%%%%%%%%%%%%%%%%%%%%%%%%%%%%%%%%%%%%%%%%%%%%%%%%%%%%%%%%%%%%%
%\subsection{Preliminary analysis}\label{sec:preliminary}
%%%%%%%%%%%%%%%%%%%%%%%%%%%%%%%%%%%%%%%%%%%%%%%%%%%%%%%%%%%%%%%%%%%%%%%

Combining these two data sets in the aforementioned event study of the 1980s, we investigate the association between the top marginal income tax rate and the right tail of the distribution of wealth.
Since wealth is cumulative, a relevant measure of tax rate is the average of lagged tax rates as opposed to the contemporaneous rate.
Figure \ref{fig:series} illustrates the time series of the 10 year average of the top marginal income tax rate in dashed lines.
For preliminary data analysis in the current section, we characterize the right tail of the distribution of wealth by estimating the Pareto exponents.
Figure \ref{fig:series} illustrates the time series of the estimated Pareto exponents in solid lines, where the left panel uses \citeauthor{Hill1975}'s (\citeyear{Hill1975}) estimator and the right panel uses \citeauthor{Gabaix2011rank}'s (\citeyear{Gabaix2011rank}) estimator.

\begin{figure}[!htb]
\centering
\begin{subfigure}{0.48\linewidth}
\includegraphics[width=\linewidth]{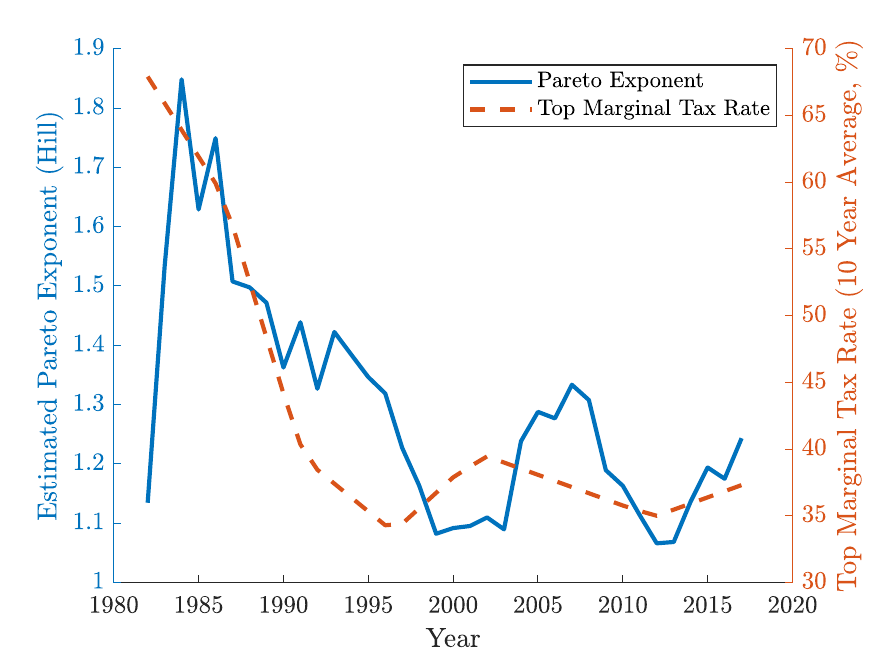}
\caption{\cite{Hill1975} estimator.}
\end{subfigure}
\begin{subfigure}{0.48\linewidth}
\includegraphics[width=\linewidth]{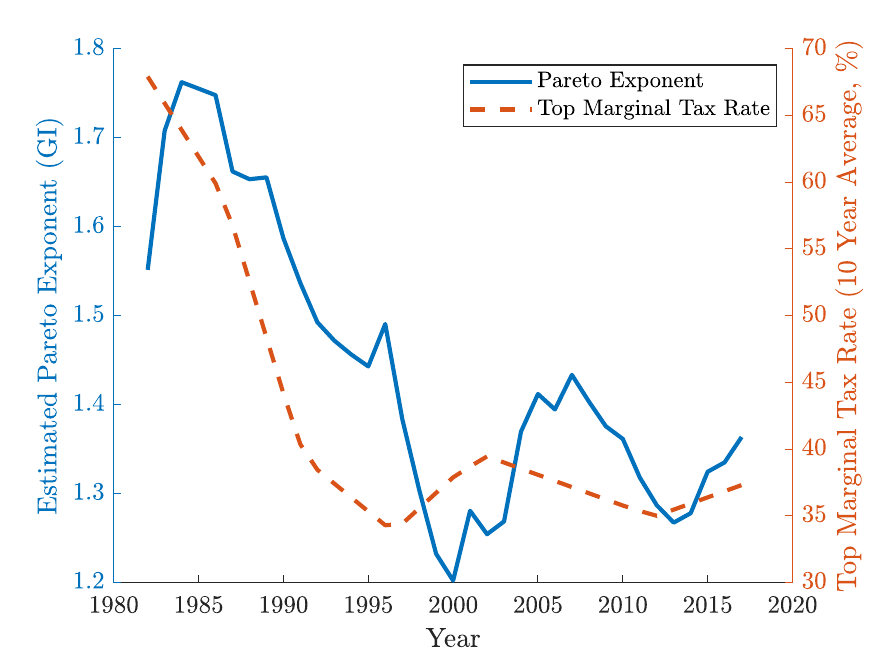}
\caption{\cite{Gabaix2011rank} estimator.}
\end{subfigure}
\caption{Time series of the top marginal income tax rate (10 year average) and the estimated Pareto exponents for the period between 1982 and 2017.
}\label{fig:series}
\end{figure}

Observe that the series of the estimated Pareto exponents follow similar trajectories to that of the top marginal income tax rate.
Namely, both the Pareto exponents and the tax rates were at high levels in the early 1980s, followed by a decade of rapid declines.
These series suggest strong positive correlations between the top marginal income tax rate and the wealth Pareto exponent (negative correlations with the heaviness of the right tail of the wealth distribution). %, and these correlations entail causal interpretations by virtue of the event of the rapid decline of the tax rate in the 1980s as a ``natural experiment'' \citep{Feldstein1995}.
The scatter plots along with fitted lines in Figure \ref{fig:scatter} highlight this strong relationship.
The data points in the left (respectively, right) panel of Figure \ref{fig:scatter} correspond to the series in the left (respectively, right) panel of Figure \ref{fig:series}.
The qualitative patterns characterizing the positive correlations are the same between the two alternative estimators of the Pareto exponents.

\begin{figure}[!htb]
\centering
\begin{subfigure}{0.48\linewidth}
\includegraphics[width=\linewidth]{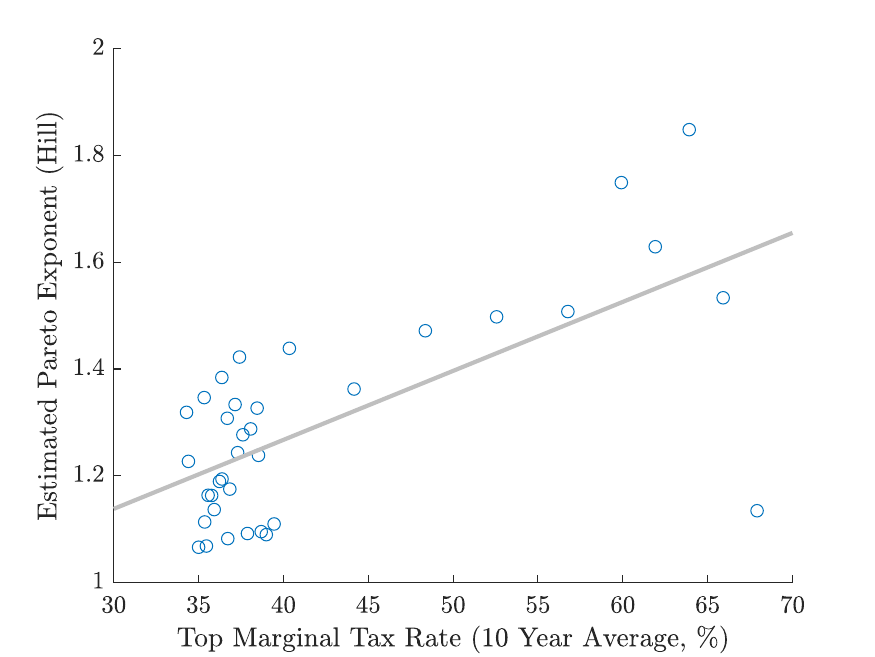}
\caption{\cite{Hill1975} estimator.}
\end{subfigure}
\begin{subfigure}{0.48\linewidth}
\includegraphics[width=\linewidth]{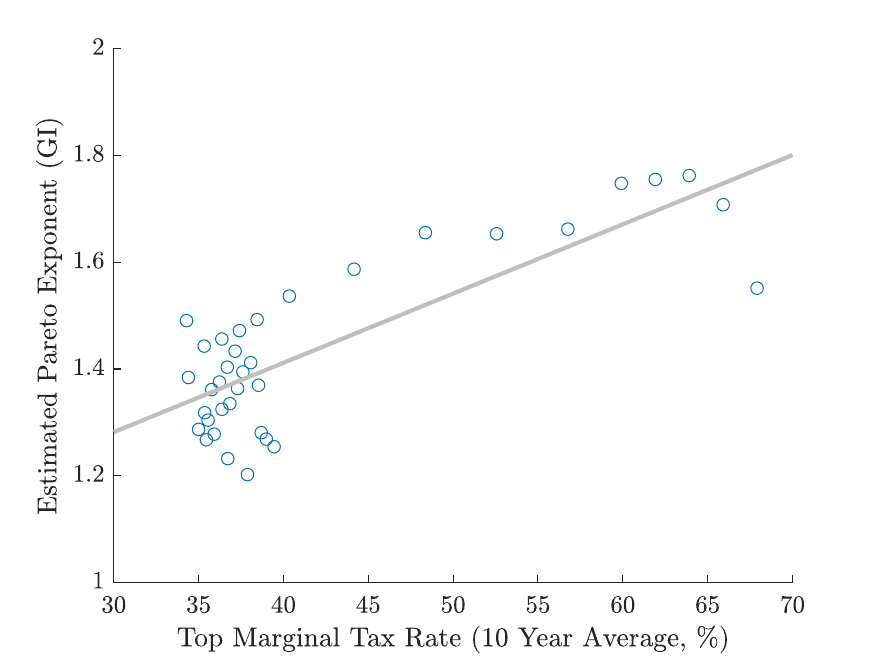}
\caption{\cite{Gabaix2011rank} estimator.}
\end{subfigure}
\caption{Scatter plots of the estimated Pareto exponents against the top marginal income tax rate (10 year average) for the period between 1982 and 2017.}
\label{fig:scatter}
\end{figure}

\section{Estimation results}\label{sec:results}
%%%%%%%%%%%%%%%%%%%%%%%%%%%%%%%%%%%%%%%%%%%%%%%%%%%%%%%%%%%%%%%%%%%%%%

In this section, we present and discuss the results of estimating the tail regression model \eqref{eq:tail_regression}. % with the data described in Section \ref{sec:data}.
We begin with formally confirming the heuristic findings from the informal bivariate analysis from Section \ref{sec:data} by first running the simple tail regression \eqref{eq:tail_regression} on only the top marginal income tax rate (and the constant) as $x$.
Since wealth is cumulative, we define $x$ as the average of the $\Delta t$ lags of the top marginal income tax rates for each of $\Delta t \in \set{3,5,7,9}$.
Columns (I)--(IV) in Table \ref{tab:estimation1} show estimates of $\theta_0$ under this tail regression model specification for various $\Delta t$.
Observe that the estimates of the coefficient of the top marginal income tax rate range quite narrowly from $-1.84$ to $-1.54$ and are significantly different from zero robustly across different time horizons $\Delta t \in \set{5,7,9}$. These results imply that the top marginal income tax rate $x$ has significantly negative effects on the tail index $\xi(x)$ of the top wealth distribution.
Recalling that the tail index is the reciprocal of the Pareto exponent, we formally confirm that an increase in the top marginal income tax rate contributes to reducing the density of wealthy population.

%%%%%%%%%%%%%%%%%%%%%%%%%%%%%%%%%%%%%%%%%%%%%%%%%%%%%%%%%%%%%%%%%%%%%%
\begin{table}[!htb]
	\centering
	\renewcommand{\arraystretch}{.75}
	\resizebox{\linewidth}{!}{%
		\begin{tabular}{lcccccccc}
		\toprule
		                     & (I) & (II) & (III) & (IV) & (V) & (VI) & (VII) & (VIII)\\
		\midrule
			Top Tax Rate       &-1.71 &-1.84$^{*}$&-1.76$^{**}$&-1.74$^{**}$&-1.75 &-1.83$^{*}$&-1.54&-1.58$^{*}$\\
			                          &(1.67)&(1.01)         &(0.89)            &(0.83)            &(1.74)&(1.04)           &(0.95)         &(0.85)\\
			\\
			Volatility         &      &      &      &      & 0.60 & 0.73 & 1.70$^{*}$ & 1.60\\
			                    &      &      &      &      &(0.97)&(1.04)&(0.99)&(1.10)\\
			Bankruptcy Rate    &      &      &      &      & 0.07 & 0.09 &8.14 &0.03\\
			                   &      &      &      &      &(8.45)&(5.62)&(29.25)&(3.35)\\
			\\
			Constant           & 1.01 & 1.07 & 1.05 & 1.06 & 0.86 & 0.86 & 0.42 & 0.54\\
			                       &(0.66)&(0.40)&(0.34)&(0.32)&(0.75)&(0.46)&(0.49)&(0.41)\\
			\\
			Years Averaged ($\Delta t$)     &    3 &    5 &    7 &    9 &    3 &    5 &    7 &    9\\
		\bottomrule
		\end{tabular}
	}
	\caption{Estimation of $\theta_0$ based on the MLE with linear specifications. ***p$<$0.01, **p$<$0.05, *p$<$0.10 (except for the constant). Standard errors are reported in parentheses. }${}$
	\label{tab:estimation1}
\end{table}
%%%%%%%%%%%%%%%%%%%%%%%%%%%%%%%%%%%%%%%%%%%%%%%%%%%%%%%%%%%%%%%%%%%%%%

We next evaluate the robustness of this observation with richer multivariate specifications of the tail regression model \eqref{eq:tail_regression} controlling for macroeconomic indicators.
Specifically, motivated by the economic model presented in Section \ref{sec:model}, we include the following additional controls to $x$:
the averages of the idiosyncratic volatility and bankruptcy rate for the $\Delta t$-year horizon for each of $\Delta t \in \set{3,5,7,9}$. To construct the idiosyncratic volatility series, we follow the method of
\cite{BekaertHodrickZhang2012}. Using NYSE stock return data from CRSP (Center
for Research in Security Prices), we calculate the daily individual excess
returns for 1,837 companies from NYSE composite. The daily individual
residual series is then constructed from Fama-French 3-factor (MKT-RF, SMB,
HML)\footnote{\url{https://mba.tuck.dartmouth.edu/pages/faculty/ken.french/data\_library.html}}
regressions. The sample variance of the residuals of $j$-th company, $\sigma^2(u_{j,y})$, for the corresponding year provides annualized
idiosyncratic variance of each company. The weighted average $\sigma_y^2=\sum_{j=1}^{N}w_{j,y}\sigma^2(u_{j,y})$ provides the
aggregate idiosyncratic variance, where the weight $w_{j,y}$ is the
fraction of the company's yearly average equity over yearly total market
equity. We then use the aggregate idiosyncratic volatility $\sigma_y$ as a regressor. We compute the annual bankruptcy rate as the yield spread of long-term AAA corporate bonds. Specifically, we obtain the monthly Moody's Seasoned AAA Corporate Bond Yield\footnote{\url{https://fred.stlouisfed.org/series/AAA}} and 20-Year Treasury Constant Maturity Rate,\footnote{\url{https://fred.stlouisfed.org/series/GS20}} compute the yield spread as their difference, and convert to annual frequency by taking the average over each year. The $\Delta t$-year averages for $\Delta t\in \set{3,5,7,9} $ are used for the empirical
analysis. Figure \ref{fig:series_covar} shows the time series of these additional covariates.

\begin{figure}[!htb]
\centering
\begin{subfigure}{0.48\linewidth}
\includegraphics[width=\linewidth]{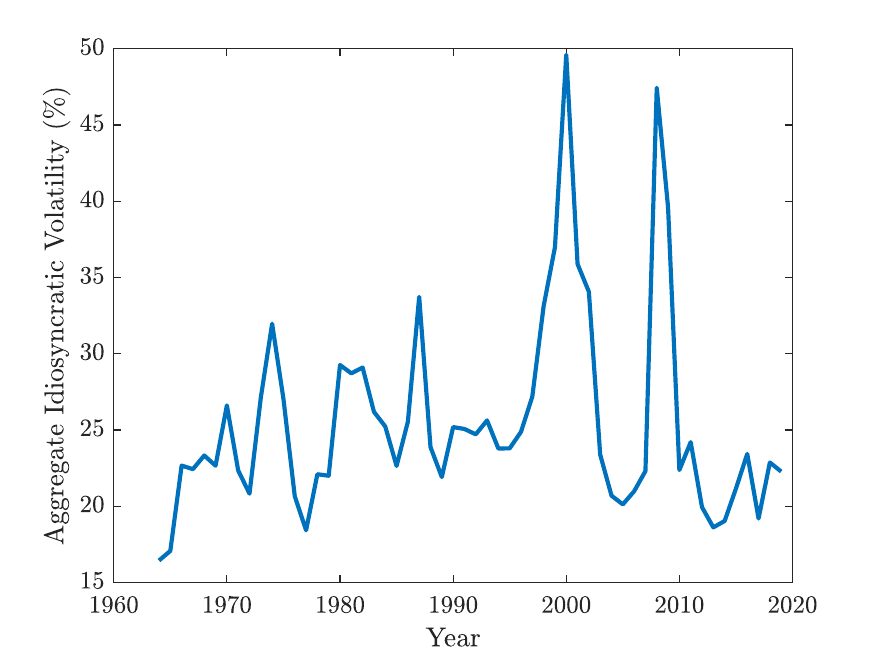}
\caption{Aggregate idiosyncratic volatility.}\label{fig:series_vol}
\end{subfigure}
\begin{subfigure}{0.48\linewidth}
\includegraphics[width=\linewidth]{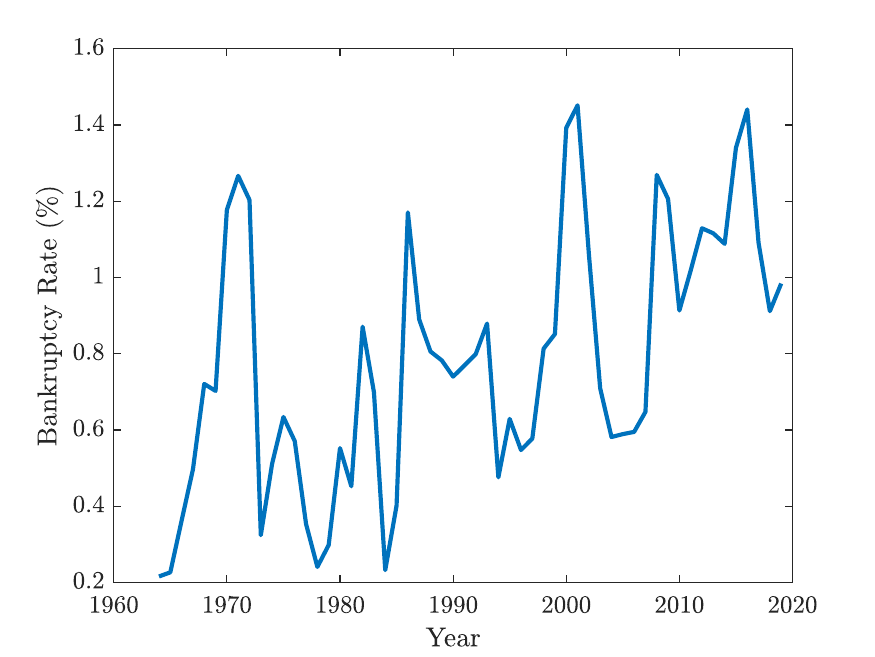}
\caption{Bankruptcy rate.}\label{fig:series_bankrupt}
\end{subfigure}
\caption{Time series of additional covariates.}
\label{fig:series_covar}
\end{figure}

Columns (V)--(VIII) in Table \ref{tab:estimation1} show estimates of $\theta_0$ under the multivariate tail regression model specification for various $\Delta t$.
Focusing on the coefficient of the top marginal income tax rate, we observe that the estimates are still close to those previously obtained without the additional covariates, and now range from $-1.84$ to $-1.54$.
These results further support the robustness of the aforementioned role that the top marginal income tax rate plays in determining the top wealth distribution.
Moveover, additional robustness check results with more control variables are presented in Appendix \ref{sec:additional} for readability.

When a policy maker revises the income tax schedule, a natural question is whether this policy instrument has any interactions (such as complementarity) with the macroeconomic state.
To address this question, we next run the multivariate tail regression \eqref{eq:tail_regression} with interactions between the top marginal income tax rate and the two controls, namely the idiosyncratic volatility and bankruptcy rate.
Table \ref{tab:estimation2} summarizes results across various combinations of interactions between the top tax rate and the two macroeconomic controls.
It turns out from these estimation results that the effects of the top marginal income tax rates on top wealth inequality are largely through the main effects, and not significantly through the interaction or complementary effects with the macroeconomic indicators.
These results suggest that a policy maker can target desired shapes of the top wealth distribution without accounting for the current macroeconomic states, as far as the two macroeconomic indicators under our consideration are concerned.

%%%%%%%%%%%%%%%%%%%%%%%%%%%%%%%%%%%%%%%%%%%%%%%%%%%%%%%%%%%%%%%%%%%%%%
\begin{table}[!htb]
	\centering
	\renewcommand{\arraystretch}{.75}
	%\resizebox{\linewidth}{!}{%
	\scalebox{0.82}{
		\begin{tabular}{lcccccccc}
		\toprule
		                     & (II) & (VI) & (IX) & (X)  & (XI) \\%& (XII) & (XIII) & (XIV)\\
		\midrule
			Top Tax Rate       &-1.84$^{*}$&-1.83$^{*}$&-2.46$^{*}$&-1.89$^{*}$&-2.45$^{*}$\\%&-21.69$^{***}$&-5.03$^{***}$&-14.14$^{*}$\\
			                         &(1.01)&(1.04)&(1.28)&(1.01)&(1.29)\\%&(7.68)&(1.57)&(7.32)\\
			\\
			Volatility         &      & 0.73 &      &      &        \\%& -106.84$^{***}$ & 1.41 & -47.15 \\
			                     &      &(1.04)&      &      &        \\%&(40.97)&(1.89)&(40.46)\\
	    Bankruptcy Rate          &      & 0.09 &      &      &        \\%& 12.10& -63.56$^{**}$ & -62.92$^{**}$ \\
			                     &      &(5.62)&      &      &        \\%&(7.81)&(28.73)&(30.46)\\
			Volatility         &      &        & 2.30 &        &2.30    \\%& 285.12$^{***}$ &      & 127.83 \\
\qquad $\times$ Tax Rate         &      &      &(2.62)&        &(2.57)  \\%&(108.30)&      &(107.76)\\
		   Bankruptcy Rate    &      &      &        &49.54  &0.09   \\%&      & 160.59$^{**}$ & 167.75$^{**}$ \\
\qquad $\times$        Tax Rate &      &      &         &(39.92)&(28.07)\\%&      &(66.67)&(69.10  )\\
			\\
			Constant       & 1.07 & 0.86 & 1.06 & 0.93 & 1.06 \\%& 8.10 & 2.27 & 5.62 \\
			                   &(0.40)&(0.75)&(0.40)&(0.39)&(0.42)\\%&(2.84)&(0.80)&(2.70)\\
			\\
			Years Averaged     &    5 &    5 &    5 &    5 &    5 \\%&    5 &    5 &    5 \\
		\bottomrule
		\end{tabular}
	}
	\caption{Estimation of $\theta_0$ based on the MLE with nonlinear specifications. ***p$<$0.01, **p$<$0.05, *p$<$0.10 (except for the constant). Standard errors are reported in parentheses. }${}$
	\label{tab:estimation2}
\end{table}
%%%%%%%%%%%%%%%%%%%%%%%%%%%%%%%%%%%%%%%%%%%%%%%%%%%%%%%%%%%%%%%%%%%%%%

Our discussions thus far are based on the estimation results for $\theta_0$ in the tail regression model \eqref{eq:tail_regression}.
A feature of this approach is that this parameter $\theta_0$ in the nonlinear single index model does not directly quantify the marginal effects of the top marginal income tax rate as a policy instrument.
The true marginal effects may depend not only on the parameter $\theta_0$, but also on the current state $x$ of the controls including the level of the top marginal income tax rate that varies over time.
To address this issue, we next use the estimator \eqref{eq:marginal} to directly measure the marginal effects for the actual policy and macroeconomic states $x$ of each of the years, 1985, 1990, 1995, 2000, 2005, 2010, and 2015.
Table \ref{tab:marginal_effects} shows the estimated marginal effects of the top marginal income tax rate for each of the eight tail regression specifications (I)--(VIII).
Despite the widely time-varying levels of $x$ and the nonlinearity in the tail regression model \eqref{eq:tail_regression}, the marginal effects remain fairly stable across time for each tail regression specification.
Specifically, increasing the top marginal income tax rate by 1 percentage point (0.01) leads to a reduction of the right tail index of the wealth distribution approximately by around 0.004 in any year (i.e., under any values $x$ of the policy and macroeconomic indicators $X_t$ in the displayed years).

%%%%%%%%%%%%%%%%%%%%%%%%%%%%%%%%%%%%%%%%%%%%%%%%%%%%%%%%%%%%%%%%%%%%%%
\begin{table}[!htb]
	\centering
	\renewcommand{\arraystretch}{.9}
	\resizebox{\linewidth}{!}{%
		\begin{tabular}{cccccccccc}
		\toprule
		& Top Tax Rate & (I) & (II) & (III) & (IV) & (V) & (VI) & (VII) & (VIII)\\
		\midrule
			1985 & 0.54 & -0.43 & -0.46$^{*}$ & -0.44$^{**}$ & -0.43$^{**}$& -0.41 & -0.46$^{*}$ & -0.39 & -0.40$^{*}$ \\
			     &           &(0.44)&(0.26)&(0.22)&(0.21)&(0.45)&(0.27)&(0.25)&(0.21) \\
			1990 & 0.35 & -0.40 & -0.44 & -0.43$^{*}$ & -0.42$^{*}$ & -0.38 & -0.44& -0.38& -0.39$^{*}$ \\
			     &           &(0.44)&(0.27)&(0.24)&(0.23)&(0.45)&(0.28)&(0.27)&(0.23) \\
			1995 & 0.36 & -0.42 & -0.44 & -0.42$^{*}$ & -0.41$^{*}$ & -0.40 & -0.44& -0.37& -0.38$^{*}$ \\
			     &           &(0.45)&(0.28)&(0.24)&(0.22)&(0.46)&(0.28)&(0.27)&(0.23) \\
			2000 & 0.40 & -0.41& -0.45& -0.43$^{*}$ & -0.42$^{*}$ & -0.39& -0.44& -0.37& -0.38$^{*}$ \\
			     &           &(0.45)&(0.28)&(0.24)&(0.22)&(0.47)&(0.29)&(0.27)&(0.23) \\
			2005 & 0.37 & -0.41& -0.44& -0.42$^{*}$ & -0.42$^{*}$ & -0.39& -0.44& -0.36 & -0.38$^{*}$ \\
			     &           &(0.45)&(0.28)&(0.24)&(0.23)&(0.46)&(0.28)&(0.27)&(0.23) \\
			2010 & 0.35 & -0.41 & -0.44& -0.42$^{*}$ & -0.41$^{*}$& -0.38& -0.43& -0.37 & -0.38$^{*}$ \\
			     &           &(0.45)&(0.28)&(0.24)&(0.22)&(0.46)&(0.28)&(0.27)&(0.23) \\
			2015 & 0.38 & -0.42 & -0.44& -0.42$^{*}$ & -0.42$^{*}$& -0.40& -0.45& -0.37 & -0.38$^{*}$ \\
			     &           &(0.45)&(0.28)&(0.24)&(0.22)&(0.46)&(0.28)&(0.27)&(0.23) \\
		\bottomrule
		\end{tabular}
	}
	\caption{Estimation of the marginal effects of the maximum marginal tax rate on the tail index. ***p$<$0.01, **p$<$0.05, *p$<$0.10. Standard errors are reported in parentheses. The tax rate indicates the 5-year average of the maximum marginal tax rates.}${}$
	\label{tab:marginal_effects}
\end{table}
%%%%%%%%%%%%%%%%%%%%%%%%%%%%%%%%%%%%%%%%%%%%%%%%%%%%%%%%%%%%%%%%%%%%%%

An important objective of economic research is to advise policy makers on concrete levels of policy instrument (top marginal tax income rate) to achieve specific objectives.
To this goal, we estimate the counterfactual values of the tail index $\xi(X_t^{\mathrm{cf}})$ under alternative values of $X_t^{\mathrm{cf}}$ under consideration by using the estimator \eqref{eq:counterfactual}.
To define the counterfactual $X_t^{\mathrm{cf}}$, we consider alternative tax rates from the list $\set{0.30, 0.40, 0.50, 0.60, 0.70}$ while we set the the other controls, namely the idiosyncratic volatility and bankruptcy rate, to the actual values in each of the years 1985, 1990, 1995, 2000, 2005, 2010, and 2015.
Table \ref{tab:counterfactual_TI} shows estimates of the actual and counterfactual tail index values for each of the alternative counterfactual policies.
Observe that the counterfactual tail index would range from 0.61 to 0.64 under the counterfactual marginal income tax rate of 0.30, while it would range from 0.43 to 0.46 under the counterfactual marginal income tax rate of 0.70.
Each of these ranges is stable across different years.
This table thus provides a guideline for a policy maker in choosing appropriate levels of the top marginal income tax rate for given goals of the density of wealthy population.

%%%%%%%%%%%%%%%%%%%%%%%%%%%%%%%%%%%%%%%%%%%%%%%%%%%%%%%%%%%%%%%%%%%%%%
\begin{table}[!htb]
	\centering
	\renewcommand{\arraystretch}{.9}
	\resizebox{\linewidth}{!}{%
		\begin{tabular}{cccccccccccccc}
		\toprule
		       & \multicolumn{2}{c}{Actual} && \multicolumn{10}{c}{Counterfactual}\\
		\cmidrule{2-3}\cmidrule{5-14}
		       & Tax  & Tail && Tax  & Tail & Tax  & Tail & Tax  & Tail & Tax  & Tail & Tax  & Tail\\
		       & Rate & Index&& Rate & Index& Rate & Index& Rate & Index& Rate & Index& Rate & Index\\
		\midrule
			1985 & 0.54 & 0.52 && 0.30 & 0.62 & 0.40 & 0.58 & 0.50 & 0.54 & 0.60 & 0.49 & 0.70 & 0.44 \\
			     &      &(0.04)&&      &(0.03)&      &(0.02)&      &(0.03)&      &(0.06)&      &(0.08)\\
			1990 & 0.35 & 0.60 && 0.30 & 0.62 & 0.40 & 0.58 & 0.50 & 0.54 & 0.60 & 0.49 & 0.70 & 0.44\\
			     &      &(0.02)&&      &(0.03)&      &(0.02)&      &(0.03)&      &(0.06)&      &(0.08)\\
			1995 & 0.36 & 0.60 && 0.30 & 0.62 & 0.40 & 0.58 & 0.50 & 0.53 & 0.60 & 0.49 & 0.70 & 0.44 \\
			     &      &(0.02)&&      &(0.03)&      &(0.02)&      &(0.03)&      &(0.06)&      &(0.08)\\
			2000 & 0.40 & 0.60 && 0.30 & 0.64 & 0.40 & 0.60 & 0.50 & 0.55 & 0.60 & 0.51 & 0.70 & 0.46 \\
			     &      &(0.03)&&      &(0.03)&      &(0.03)&      &(0.04)&      &(0.06)&      &(0.08)\\
			2005 & 0.37 & 0.60 && 0.30 & 0.63 & 0.40 & 0.58 & 0.50 & 0.54 & 0.60 & 0.49 & 0.70 & 0.45  \\
			     &      &(0.02)&&      &(0.03)&      &(0.02)&      &(0.03)&      &(0.06)&      &(0.08)\\
			2010 & 0.35 & 0.61 && 0.30 & 0.63 & 0.40 & 0.59 & 0.50 & 0.55 & 0.60 & 0.50 & 0.70 & 0.46 \\
			     &      &(0.02)&&      &(0.03)&      &(0.02)&      &(0.03)&      &(0.06)&      &(0.08)\\
			2015 & 0.38 & 0.58 && 0.30 & 0.61 & 0.40 & 0.57 & 0.50 & 0.53 & 0.60 & 0.48 & 0.70 & 0.43 \\
			     &      &(0.02)&&      &(0.03)&      &(0.02)&      &(0.04)&      &(0.06)&      &(0.08)\\
		\bottomrule
		\end{tabular}
	}
	\caption{Estimation of the actual and counterfactual tail indices under Model (VI). Standard errors are reported in parentheses. The actual and counterfactual tax rate indicates the 5-year average of the maximum marginal tax rates.}${}$
	\label{tab:counterfactual_TI}
\end{table}
%%%%%%%%%%%%%%%%%%%%%%%%%%%%%%%%%%%%%%%%%%%%%%%%%%%%%%%%%%%%%%%%%%%%%%

The tail regression \eqref{eq:tail_regression} is modeled in terms of the tail index $\xi(x)$ as the response variable because the tail index can be conveniently restricted to the interval $(0,1)$, consistently with the range of the link function $\Lambda$ in \eqref{eq:tail_regression}.
That said, its reciprocal, namely the Pareto exponent $\alpha=\alpha(x)=1/\xi(x)$, is the measure of tail thickness that has been used more often in the economics literature.
To translate the counterfactual tail index values reported in Table \ref{tab:counterfactual_TI} into this more familiar measure, we use the estimator \eqref{eq:counterfactual_pareto} and report the counterfactual Pareto exponents in Table \ref{tab:counterfactual_PE}.
The maximum level, 0.70, of the counterfactual marginal income tax rate under our consideration used to be the steady policy state for long prior to the 1980s in the United States, while our Forbes 400 data are available only since 1982.
Our extrapolated estimates of the counterfactual Pareto exponent, ranging from 2.17 to 2.31, thus recovers the density of the wealthy population in the United States under the high-income-tax regime before the 1980s that is not directly available from our data.
These estimates also suggest the density of wealthy population that our economy would achieve if the maximum marginal income tax rate were to be raised again back to old steady level of 0.70.

%%%%%%%%%%%%%%%%%%%%%%%%%%%%%%%%%%%%%%%%%%%%%%%%%%%%%%%%%%%%%%%%%%%%%%
\begin{table}[!htb]
	\centering
	\renewcommand{\arraystretch}{.9}
	\resizebox{\linewidth}{!}{
		\begin{tabular}{cccccccccccccc}
		\toprule
		       & \multicolumn{2}{c}{Actual} && \multicolumn{10}{c}{Counterfactual}\\
		\cmidrule{2-3}\cmidrule{5-14}
		       & Tax  & Pareto && Tax  & Pareto & Tax  & Pareto & Tax  & Pareto & Tax  & Pareto & Tax  & Pareto\\
		       & Rate & Exp.&& Rate & Exp.& Rate & Exp.& Rate & Exp.& Rate & Exp.& Rate & Exp.\\
		\midrule
			1985 & 0.54 & 1.93 && 0.30 & 1.60 & 0.40 & 1.72 & 0.50 & 1.87 & 0.60 & 2.04 & 0.70 & 2.25 \\
			     &      &(0.16)&&      &(0.07)&      &(0.05)&      &(0.12)&      &(0.24)&      &(0.41)\\
			1990 & 0.35 & 1.66 && 0.30 & 1.60 & 0.40 & 1.72 & 0.50 & 1.87 & 0.60 & 2.05 & 0.70 & 2.25\\
			     &      &(0.05)&&      &(0.07)&      &(0.05)&      &(0.12)&      &(0.24)&      &(0.41)\\
			1995 & 0.36 & 1.68 && 0.30 & 1.61 & 0.40 & 1.73 & 0.50 & 1.88 & 0.60 & 2.06 & 0.70 & 2.27 \\
			     &      &(0.05)&&      &(0.07)&      &(0.05)&      &(0.12)&      &(0.25)&      &(0.42)\\
			2000 & 0.40 & 1.68 && 0.30 & 1.56 & 0.40 & 1.67 & 0.13 & 1.81 & 0.60 & 1.98 & 0.70 & 2.17 \\
			     &      &(0.08)&&      &(0.08)&      &(0.08)&      &(0.13)&      &(0.23)&      &(0.39)\\
			2005 & 0.37 & 1.68 && 0.30 & 1.60 & 0.40 & 1.72 & 0.50 & 1.86 & 0.60 & 2.03 & 0.70 & 2.24  \\
			     &      &(0.05)&&      &(0.07)&      &(0.05)&      &(0.11)&      &(0.24)&      &(0.41)\\
			2010 & 0.35 & 1.63 && 0.30 & 1.57 & 0.40 & 1.69 & 0.50 & 1.83 & 0.60 & 2.00 & 0.70 & 2.20 \\
			     &      &(0.06)&&      &(0.07)&      &(0.06)&      &(0.12)&      &(0.23)&      &(0.40)\\
			2015 & 0.38 & 1.73 && 0.30 & 1.63 & 0.40 & 1.75 & 0.50 & 1.90 & 0.60 & 2.09 & 0.70 & 2.31 \\
			     &      &(0.07)&&      &(0.08)&      &(0.07)&      &(0.14)&      &(0.27)&      &(0.45)\\
		\bottomrule
		\end{tabular}
	}
	\caption{Estimation of the actual and counterfactual Pareto exponents under Model (VI). Standard errors are reported in parentheses. The actual and counterfactual tax rate indicates the 5-year average of the maximum marginal tax rates.}${}$
	\label{tab:counterfactual_PE}
\end{table}
%%%%%%%%%%%%%%%%%%%%%%%%%%%%%%%%%%%%%%%%%%%%%%%%%%%%%%%%%%%%%%%%%%%%%%

Finally, we conduct out-of-sample counterfactual predictions for the Pareto exponents in year 2020 using the controls that are available up to year 2019.
The top row in Table \ref{tab:counterfactual_PE_prediction} shows estimates of the Pareto exponents in 2020 under each of the alternative counterfactual policies of the maximum marginal income tax rates from the list $\set{0.20,0.30,\dots,0.70,0.80}$.
As expected, a higher marginal tax rate will increase the Pareto exponent, implying a thinner tail and hence less inequality. 
Observe that the counterfactual maximum marginal income tax rate of 0.60 or above would yield the Pareto exponent above 2 (in terms of the point estimates), allowing for finite second moments of the wealth distribution.

%%%%%%%%%%%%%%%%%%%%%%%%%%%%%%%%%%%%%%%%%%%%%%%%%%%%%%%%%%%%%%%%%%%%%%
\begin{table}[!htb]
	\centering
	\renewcommand{\arraystretch}{.75}
	\resizebox{\linewidth}{!}{%
		\begin{tabular}{lccccccc}
		\toprule
			Counterfactual Tax Rate & 0.20 & 0.30 & 0.40 & 0.50 & 0.60 & 0.70 & 0.80\\
		\midrule
			Pareto Exponent & 1.52 & 1.63 & 1.75 & 1.91 & 2.09 & 2.31 & 2.57\\
			                        &(0.11)&(0.08)&(0.07)&(0.14)&(0.27)&(0.45)&(0.70)\\
			\\
			Test of Finite Moments:\\
			First Moment & $<\infty$ & $<\infty$ & $<\infty$ & $<\infty$ & $<\infty$ & $<\infty$ & $<\infty$\\
			                    &(1.00)&(1.00)&(1.00)&(1.00)&(1.00)&(1.00)&(0.99)\\
			Second Moment & $=\infty$ & $=\infty$  & $=\infty$ & $<\infty$ &$<\infty$& $<\infty$& $<\infty$\\
			                    &(0.00)&(0.00)&(0.00)&(0.26)&(0.63)&(0.75)&(0.79)\\
			Third Moment & $=\infty$ & $=\infty$ & $=\infty$ & $=\infty$ & $=\infty$ & $=\infty$ & $<\infty$\\
			                    &(0.00)&(0.00)&(0.00)&(0.00)&(0.00)&(0.06)&(0.27)\\
		\bottomrule
		\end{tabular}
	}
	\caption{Extrapolated estimates of the counterfactual Pareto exponents in 2020 under Model (VI) and tests of finite first, second, and third moments of the wealth distribution. Standard errors are reported in parentheses below the estimates. P values of the tests of infinite moments are reported in parentheses below the test results. The counterfactual tax rate indicates the 5-year average of the maximum marginal tax rates.}${}$
	\label{tab:counterfactual_PE_prediction}
\end{table}
%%%%%%%%%%%%%%%%%%%%%%%%%%%%%%%%%%%%%%%%%%%%%%%%%%%%%%%%%%%%%%%%%%%%%%

To conduct a formal test of this statement, however, we should also pay attention to the standard errors as well as the point estimates.
Using these estimates and the the associated standard errors, we conduct the one-sided tests of the null hypothesis of the finite $r$-th moment of the wealth distribution for $r \in \set{1,2,3}$.
The bottom parts of Table \ref{tab:counterfactual_PE_prediction} show the test results.
We mean by ``$<\infty$'' that we fail to reject the null hypothesis of $\alpha(X_t^{\mathrm{cf}}) \ge r$ at 5\% significance level, while we mean by ``$=\infty$'' that we reject it.
The p-values are reported in parentheses.
The test results suggest that the first moment of the wealth distribution is finite at any counterfactual tax rate, but the second and higher moments are infinite at lower tax rates as is the case with the United States today.
In order to reduce the wealth inequality, the tax rate needs to be raised to 0.60 for finite second moment and 0.80 or above for finite third moment.
%Recall that a finite second moment is typically required and taken for granted in empirical anlaysis for the validity of the t-test. 
%The heavy tail raises such concern whenever wealth data is used.
%See, for example, \citet{SasakiWang2021}.

%%%%%%%%%%%%%%%%%%%%%%%%%%%%%%%%%%%%%%%%%%%%%%%%%%%%%%%%%%%%%%%%%%%%%%
\section{A model of tax rate and Pareto exponent}\label{sec:model}
%%%%%%%%%%%%%%%%%%%%%%%%%%%%%%%%%%%%%%%%%%%%%%%%%%%%%%%%%%%%%%%%%%%%%%

In the previous two sections, we have documented positive effects of the marginal tax rate on the wealth Pareto exponent. In this section, we present a minimal stylized theoretical model that sheds light on these effects.

\subsection{Model description}
Because a large fraction of Forbes 400 individuals are self-made entrepreneurs,\footnote{According to Forbes, about 70\% of individuals in the list are self-made: \url{https://www.forbes.com/sites/jonathanponciano/2020/09/08/self-made-score}.} we model the top of the wealth distribution as consisting of entrepreneurs. We consider an economy consisting of two agent types, entrepreneurs and workers, with mass $1$ and $L$, respectively. Time is continuous and is denoted by $t\in [0,\infty)$. Workers inelastically supply one unit of labor and earn wage $\omega_t$, which they consume entirely.\footnote{This hand-to-mouth assumption is for simplicity only and can be microfounded by supposing that workers are impatient and cannot borrow.}

Entrepreneurs operate their own firms using capital and labor as inputs and trading risk-free bonds in zero net supply, pay taxes on profits, and go bankrupt at a Poisson rate $\lambda>0$. The budget constraint of a typical entrepreneur is
\begin{equation}
\diff w_t=(1-\tau_t)((F(k_t,l_t)-\omega_t l_t)\diff t + \sigma k_t\diff B_t + r_tb_t\diff t) -c_t\diff t.\label{eq:bc}
\end{equation}
Here $k_t,l_t$ are capital and labor inputs and $F$ is the production function including capital depreciation; $\tau_t\in [0,1)$ is the marginal income tax rate; $B_t$ is a standard Brownian motion assumed to be independent across entrepreneurs and $\sigma>0$ is the idiosyncratic volatility; $r_t$ is the interest rate including the risk premium for bankruptcy and $b_t$ is bond holdings; $c_t$ is consumption rate; and $w_t=k_t+b_t$ is net worth. When entrepreneurs go bankrupt, they exit the economy and are replaced by new entrepreneurs. Upon exit, capital is recycled back to the economy, and the initial endowment of a new entrepreneur equals the average capital of exiting entrepreneurs. The tax revenue is used in a way that does not affect the behavior of entrepreneurs (e.g., wasted or redistributed to workers).

Entrepreneurs have the continuous-time analog of the Epstein-Zin preferences with discount rate $\beta$, relative risk aversion $\gamma$, and elasticity of intertemporal substitution equal to 1. More precisely, the continuation utility $U_t$ satisfies
\begin{equation}
U_t=\E_t\int_t^\infty h(c_s,U_s)\diff s,\label{eq:utility}
\end{equation}
where $c_t$ is the consumption rate at time $t$ and
\begin{equation}
h(c,v)=\beta v\left((1-\gamma)\log c - \log [(1-\gamma)v]\right) \label{eq:hcv1}
\end{equation}
is the intertemporal aggregator. (See \cite{duffie-epstein1992a} for technical details.) The discount rate $\beta$ should be understood to incorporate the agent's preference for current over future consumption as well as their awareness of the risk of bankruptcy. The objective of an entrepreneur is to maximize the recursive utility \eqref{eq:utility} subject to the budget constraint \eqref{eq:bc} and nonnegativity constraints $k_t,l_t,c_t\ge 0$, taking as given the paths of interest rate and wage $\set{(r_t,\omega_t)}_{t\in [0,\infty)}$.

Throughout the rest of this section, we maintain the following assumption.

\begin{asmp}\label{asmp:prod}
The production function $F(k,l)$ is homogeneous of degree 1. Letting $f(k)\coloneqq F(k,1)$ be the production function with unit labor, $f:(0,\infty)\to \R$ is twice continuously differentiable and satisfies $f''<0$, $f'(0)=\infty$, and $f'(\infty)\le 0$.
\end{asmp}

A typical example satisfying Assumption \ref{asmp:prod} is the Cobb-Douglas production function with constant depreciation
\begin{equation}
F(k,l)=Ak^\rho l^{1-\rho}-\delta k, \label{eq:CD}
\end{equation}
where $A>0$ is the productivity, $\rho\in (0,1)$ is the capital share, and $\delta\ge 0$ is the depreciation rate.

\subsection{Entrepreneur's problem}

In this section we solve the entrepreneur's problem. We suppress the time subscript to simplify the notation. Since labor $l$ enters only the budget constraint \eqref{eq:bc}, given the wage $\omega$ and capital holdings, it is obvious that the entrepreneur chooses $l$ to solve
\begin{equation}
\max_{l\ge 0} [F(k,l)-\omega l]=\max_{l\ge 0} [lf(k/l)-\omega l], \label{eq:lprob}
\end{equation}
where we have used the homogeneity of $F$ and the definition of $f$. Using $f''<0$, it is straightforward to show that the objective function in the right-hand side of \eqref{eq:lprob} is strictly concave. The first-order condition with respect to $l$ is
\begin{equation}
f(y)-yf'(y)=\omega,\label{eq:lfoc}
\end{equation}
where $y=k/l$ is the capital-labor ratio. Since $(f(y)-yf'(y))'=-yf''(y)>0$, the value of $y$ (hence $l$) satisfying \eqref{eq:lfoc} is unique. Since the wage $\omega$ and the capital-labor ratio $y$ have a one-to-one relationship, in the subsequent discussion we use $y$ as an endogenous variable and recover $\omega$ using the first-order condition \eqref{eq:lfoc}. Given capital $k$ and capital-labor ratio $y=k/l$, we have
\begin{equation*}
F(k,l)-\omega l=l(f(k/l)-\omega)=l(f(y)-\omega)=lyf'(y)=kf'(y),
\end{equation*}
where we have used the first-order condition \eqref{eq:lfoc}. Therefore the budget constraint \eqref{eq:bc} reduces to
\begin{equation}
\diff w=[(1-\tau)(f'(y)\theta+r(1-\theta))-m]w\diff t+(1-\tau)\sigma\theta w\diff B,\label{eq:bc2}
\end{equation}
where $\theta\coloneqq k/w\ge 0$ is the fraction of wealth invested in capital and $m\coloneqq c/w$ is the marginal propensity to consume. The following proposition characterizes the solution to the entrepreneur's problem.

\begin{proposition}\label{prop:agent1}
Consider an entrepreneur facing expected return $\mu_t=f'(y_t)$, tax rate $\tau_t$, and volatility $\sigma$. Then the solution to the entrepreneur's problem is
\begin{subequations}
\begin{align}
\theta_t&=\frac{\max\set{\mu_t-r_t,0}}{(1-\tau_t)\gamma\sigma^2},\label{eq:optTheta1}\\
m_t&=\beta.\label{eq:optC1}
\end{align}
\end{subequations}
\end{proposition}

\begin{proof}
Special case of Proposition \ref{prop:agent} in Appendix \ref{sec:proof}.
\end{proof}

\subsection{Equilibrium}

We now characterize the equilibrium and the stationary wealth distribution. Loosely speaking, given the path of tax rates $\set{\tau_t}_{t\in [0,\infty)}$ and the initial condition (wealth distribution), a sequential equilibrium is defined by paths of interest rate and wage $\set{(r_t,\omega_t)}_{t\in [0,\infty)}$, consumption-portfolio rules of entrepreneurs, and wealth distributions such that entrepreneurs maximize utility, the bond and labor markets clear, and the wealth distribution is consistent with the behavior of entrepreneurs and the entry/exit mechanism.

Due to our simplifying assumption of unit elasticity, the equilibrium analysis is quite tractable. First, note that the optimal behavior of entrepreneurs is characterized as in Proposition \ref{prop:agent1}. Because entrepreneurs hold identical portfolios $\theta$ in \eqref{eq:optTheta1}, the bond is in zero net supply, and only entrepreneurs trade bonds, the bond market clearing condition implies
\begin{equation}
1-\theta=0\iff \theta=1.\label{eq:bclear}
\end{equation}
%Then from \eqref{eq:optTheta1}, the expected return must satisfy
%\begin{equation}
%\mu=r+(1-\tau)\gamma\sigma^2.\label{eq:mu}
%\end{equation}
By the first-order condition of the profit maximization \eqref{eq:lfoc}, all entrepreneurs choose the same capital-labor ratio $y=k/l$. Because aggregate labor supply is fixed at $L$, letting $K$ be the aggregate capital, the labor market clearing condition implies
\begin{equation}
y=k/l=K/L.\label{eq:Lclear}
\end{equation}
%Combining \eqref{eq:mu}, \eqref{eq:Lclear}, and $\mu=f'(y)$, the interest rate satisfies
%\begin{equation}
%r=f'(K/L)-(1-\tau)\gamma\sigma^2.\label{eq:rK}
%\end{equation}

Next, substituting \eqref{eq:optC1}, \eqref{eq:bclear}, and \eqref{eq:Lclear} into the budget constraint \eqref{eq:bc2} and using $b=0$, we obtain the individual wealth (capital) dynamics
\begin{equation}
\diff k =gk\diff t+vk\diff B \coloneqq [(1-\tau)f'(K/L)-\beta]k\diff t+(1-\tau)\sigma k\diff B.\label{eq:kdynamics}
\end{equation}
Because the capital of exiting entrepreneurs is recycled back to the economy as the initial capital of entering entrepreneurs and the investment risk $\diff B$ is independent across entrepreneurs, aggregating \eqref{eq:kdynamics} we obtain the law of motion for aggregate capital
\begin{equation}
\diff K=[(1-\tau)f'(K/L)-\beta]K\diff t.\label{eq:Kdynamics}
\end{equation}
In the steady state ($\tau$ and $K$ are constant), \eqref{eq:Kdynamics} implies that aggregate capital satisfies
\begin{equation}
(1-\tau)f'(K/L)-\beta=0\iff f'(K/L)=\frac{\beta}{1-\tau}.\label{eq:steadyK}
\end{equation}
%Combining \eqref{eq:rK} and \eqref{eq:steadyK}, the steady state interest rate is
%\begin{equation}
%r=\frac{\beta}{1-\tau}-(1-\tau)\gamma\sigma^2. \label{eq:steadyr}
%\end{equation}

Because the wealth dynamics \eqref{eq:kdynamics} of entrepreneurs is a geometric Brownian motion (when $K$ is constant), and entrepreneurs enter and exit at Poisson rate $\lambda$, it follows from the result of \cite{reed2001} that the stationary wealth distribution is double Pareto, which obeys the power law in both the upper and lower tails.\footnote{This ``double power law'' is empirically observed for income as documented by \cite{reed2001} and \cite{Toda2012JEBO}. More generally, \cite{BeareToda2022ECMA} (discrete-time) and \cite{BeareSeoToda2021ET} (continuous-time) show that a random Markov-modulated multiplicative process stopped at an exponentially-distributed time exhibits Pareto upper and lower tails. Although we focus on the \iid case for simplicity, it is straightforward to allow for Markov modulation in our model and apply their results to characterize the Pareto exponent.} According to \citet[Equation (A.4)]{reed2001}, the Pareto exponent $\alpha$ is the positive root of the quadratic equation
\begin{equation}
\frac{v^2}{2}z^2+\left(g-\frac{v^2}{2}\right)z-\lambda=0,\label{eq:quad}
\end{equation}
where $g$ is the expected growth rate and $v$ is the volatility of the geometric Brownian motion. We therefore obtain the following proposition.

\begin{proposition}\label{prop:PE}
The stationary wealth distribution of entrepreneurs is double Pareto. The Pareto exponent is given by
\begin{equation}
\alpha=\frac{1}{2}\left(1+\sqrt{1+\frac{8\lambda}{(1-\tau)^2\sigma^2}}\right),\label{eq:PE}
\end{equation}
which is increasing in the marginal tax rate $\tau$.
\end{proposition}

\begin{proof}
Since in the steady state we have $g=0$ and $v=(1-\tau)\sigma$ in \eqref{eq:kdynamics}, solving the quadratic equation \eqref{eq:quad} for the positive root, we obtain \eqref{eq:PE}. Since $\tau\in [0,1)$ and all parameters are positive, $\alpha$ is increasing in $\tau$.
\end{proof}

What is remarkable about Proposition \ref{prop:PE} is that the theoretical Pareto exponent $\alpha$ takes the simple form \eqref{eq:PE}, which depends only on the idiosyncratic volatility, bankruptcy rate, and the tax rate. In particular, it is independent from preference parameters (discount rate, risk aversion, etc.), the production function, and aggregate labor supply. As a numerical illustration, we set $\sigma=0.2633$ and $\lambda=0.0078$, which are the average values in our data set used in Section \ref{sec:results}. Figure \ref{fig:scatterModel} shows the theoretical Pareto exponent computed from \eqref{eq:PE} when we vary the marginal tax rate in the range between 30\% and 70\%, together with the empirical estimates obtained in Figure \ref{fig:scatter}. Although the model is quite stylized, the theoretical Pareto exponent is similar to the empirical estimates.

\begin{figure}[!htb]
\centering
\begin{subfigure}{0.48\linewidth}
\includegraphics[width=\linewidth]{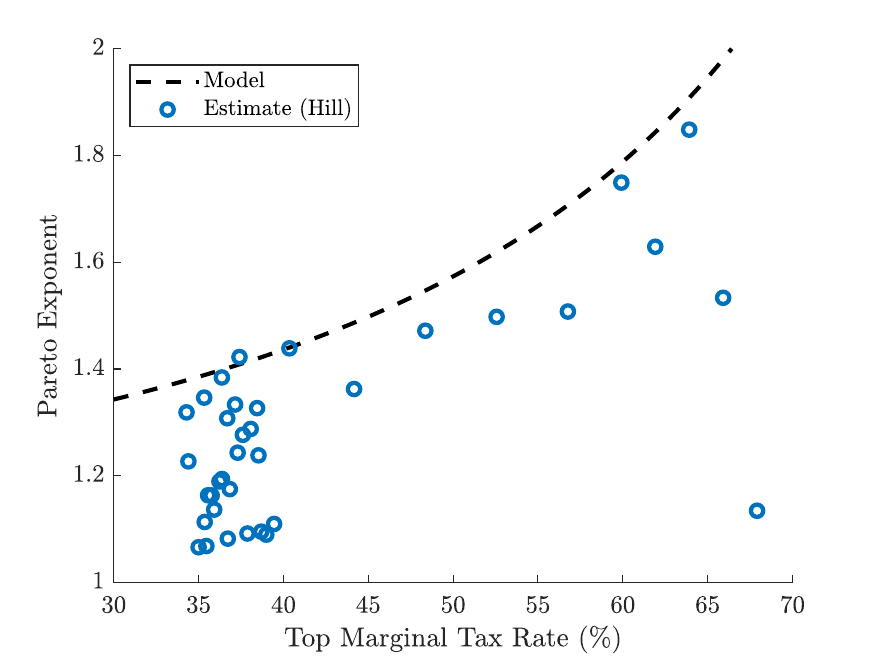}
\caption{\cite{Hill1975} estimator.}
\end{subfigure}
\begin{subfigure}{0.48\linewidth}
\includegraphics[width=\linewidth]{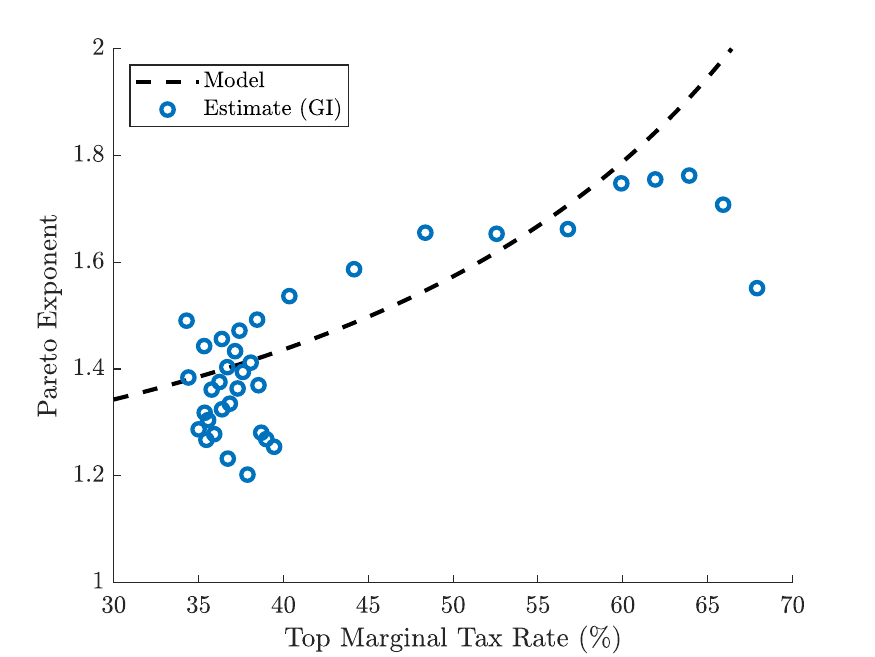}
\caption{\cite{Gabaix2011rank} estimator.}
\end{subfigure}
\caption{Marginal tax rate and Pareto exponent.}
\label{fig:scatterModel}
\end{figure}

\subsection{Taxation and welfare}

So far we know from Proposition \ref{prop:PE} that within the model, a higher tax rate unambiguously reduces wealth inequality. However, inequality itself may not be a relevant policy target. We now use the model to study the welfare implications of taxation.

Because the economy consists of two agent types, instead of considering a particular social welfare function, we study the welfare of entrepreneurs and workers separately. For entrepreneurs, according to Proposition \ref{prop:agent}, the value function of an entrepreneur with wealth $w$ takes the form $J(w)=\frac{1}{1-\gamma}(aw)^{1-\gamma}$ for some constant $a>0$ that depends on exogenous parameters. Because an entrepreneur has initial capital $K$, we convert the value function into units of consumption and define the welfare criterion by $\cW_E=aK$. For workers, because they are hand-to-mouth, the welfare criterion is consumption. If the tax revenue is wasted, then the worker welfare equals the wage and $\cW_W^0\coloneqq \omega$. If the tax revenue is fully redistributed to workers, then worker welfare is $\cW_W^1\coloneqq \omega+\cT/L$, where $\cT$ is the aggregate tax revenue. The following proposition characterizes these objects.

\begin{proposition}\label{prop:welfare}
Let
\begin{equation}
K=K(\tau)\coloneqq (f')^{-1}\left(\frac{\beta}{1-\tau}\right)L \label{eq:Ktau}
\end{equation}
be the steady state aggregate capital determined by \eqref{eq:steadyK}. Then the welfare of workers and entrepreneurs as well as the aggregate tax revenue are given by
\begin{subequations}\label{eq:welfare}
\begin{align}
\cW_W^0&=f(K/L)-(K/L)f'(K/L),\label{eq:WW0}\\
\cW_W^1&=f(K/L)-\beta K/L,\label{eq:WW1}\\
\cW_E&=\beta \e^{-\frac{1}{2\beta}(1-\tau)^2\gamma\sigma^2}K,\label{eq:WE}\\
\cT&=\tau Kf'(K/L)=(f'(K/L)-\beta)K.\label{eq:T}
\end{align}
\end{subequations}
The worker welfare with/without redistribution $\cW_W^1,\cW_W^0$ are both strictly decreasing in the tax rate $\tau$. If in addition the production function is Cobb-Douglas as in \eqref{eq:CD}, then the entrepreneur welfare $\cW_E$ and tax revenue $\cT$ are both single-peaked (initially increasing and then decreasing) in $\tau$.
\end{proposition}

The intuition for why the worker welfare is decreasing in the tax rate is straightforward: when the tax rate is high, the steady state capital is low due to the concavity of the production function (see \eqref{eq:steadyK}), and hence the wage is low. The intuition for why the entrepreneur welfare is non-monotonic is due to the trade-off between risk and wealth level. A higher tax rate provides insurance against capital income risk because losses can be deducted, which improves welfare as we can see from the first term in \eqref{eq:WE}. However, a higher tax rate also prevents capital accumulation and hence wealth, as we can see from \eqref{eq:steadyK} and the second term in \eqref{eq:WE}. Consequently, entrepreneurs prefer an intermediate tax rate.

As a numerical illustration, we set $\beta=0.05$, $\gamma=2$, $\rho=0.38$, and $\delta=0.08$, which are all standard values. We also set $A=1$ (normalization) and $L=9$ (so 90\% of agents are workers), although these two parameters are unimportant due to the homogeneity of the production function (different choices of $(A,L)$ only shift the quantities in \eqref{eq:welfare} by some multiplicative factors independent of $\tau$). Figure \ref{fig:welfare} shows the worker welfare (with/without redistribution), entrepreneur welfare, and the tax revenue. Consistent with Proposition \ref{prop:welfare}, the worker welfare in Figure \ref{fig:ModelWW} is decreasing in the tax rate $\tau$, although with redistribution the welfare is nearly flat for tax rates below 40\%. The entrepreneur welfare in Figure \ref{fig:ModelWE} is greatly affected by the tax rate and peaks at $\tau=0.444$, which is only slightly higher than the current top marginal income tax rate (37\%). The tax revenue in Figure \ref{fig:ModelT} is single-peaked, consistent with Proposition \ref{prop:welfare} and the Laffer curve.

\begin{figure}[!htb]
\centering
\begin{subfigure}{0.48\linewidth}
\includegraphics[width=\linewidth]{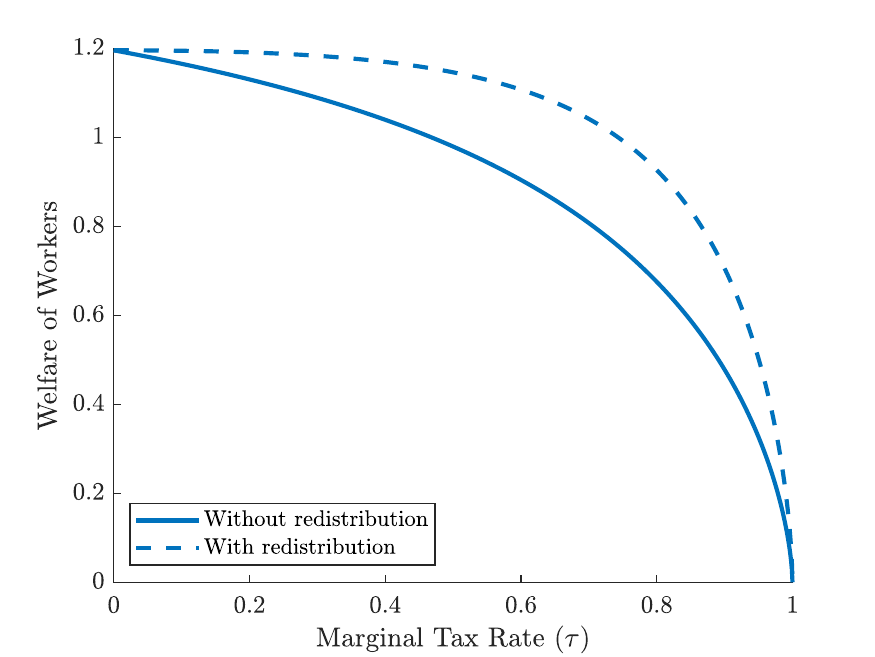}
\caption{Worker welfare.}\label{fig:ModelWW}
\end{subfigure}
\begin{subfigure}{0.48\linewidth}
\includegraphics[width=\linewidth]{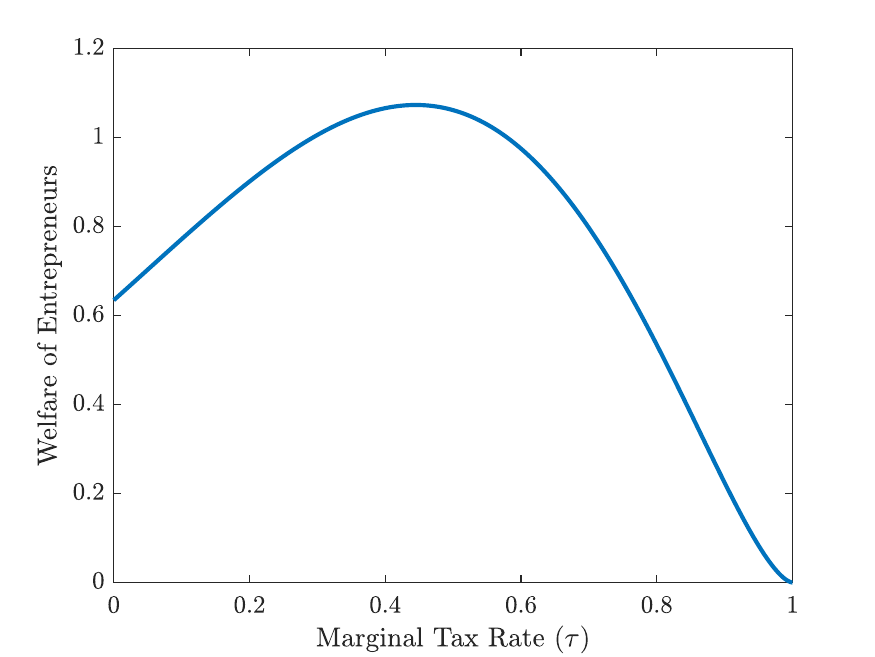}
\caption{Entrepreneur welfare.}\label{fig:ModelWE}
\end{subfigure}
\begin{subfigure}{0.48\linewidth}
\includegraphics[width=\linewidth]{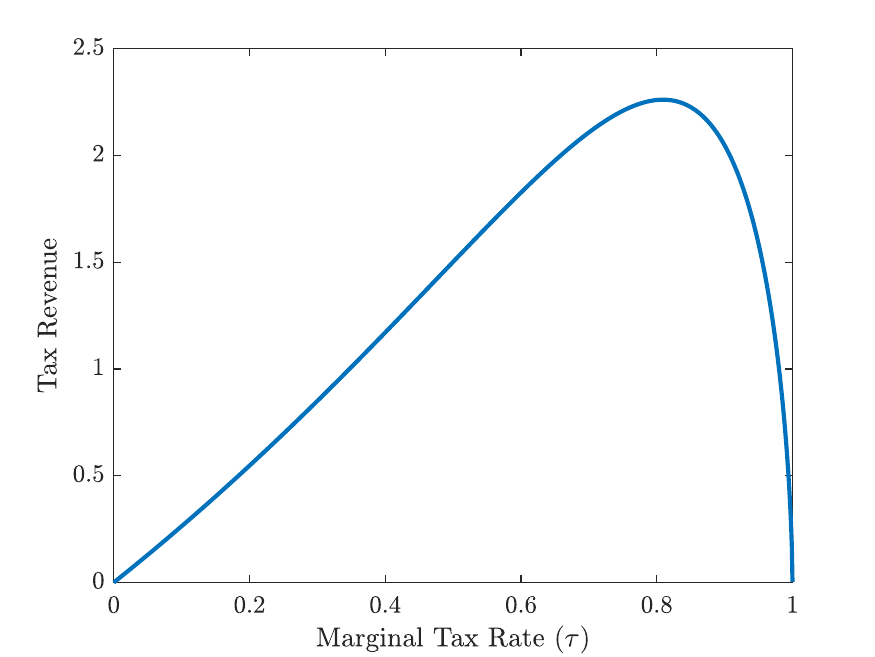}
\caption{Tax revenue.}\label{fig:ModelT}
\end{subfigure}
\caption{Effects of tax rate on welfare and tax revenue.}
\label{fig:welfare}
\end{figure}

%%%%%%%%%%%%%%%%%%%%%%%%%%%%%%%%%%%%%%%%%%%%%%%%%%%%%%%%%%%%%%%%%%%%%%
\section{Summary and discussions}\label{sec:summary} 
%%%%%%%%%%%%%%%%%%%%%%%%%%%%%%%%%%%%%%%%%%%%%%%%%%%%%%%%%%%%%%%%%%%%%%

We use Forbes 400 data along with historical data on tax rates and macroeconomic indicators to study the effects of the maximum marginal income tax rate on wealth inequality.
To analyze these effects, we propose a novel truncated tail regression model and an econometric method to analyze this model.
Using these data sets and the econometric model,
% in the ``natural experiment'' \citep{Feldstein1995} of the event that the tax rate rapidly declined in the 1980s, 
we find that a higher maximum tax rate induces a lower tail index, or equivalently, a higher Pareto exponent of the wealth distribution.
Setting the maximum tax rate to 30--40\%, as is the case with U.S. today, leads to the Pareto exponent of 1.5--1.8. %, entailing an infinite variance of the wealth distribution.
%To ensure a wealth distribution with finite second moment, the tax rate needs to be raised to 0.60 or above.
%In particular,
Counterfactually setting the maximum tax rate to 80\% as suggested by \citet{piketty2013capital} would achieve the Pareto exponent of 2.6.
Finally, we develop a simple economic model that explains these empirical findings about the effects of the marginal tax rate on wealth inequality. 
We find that the worker welfare is strictly decreasing in the tax rate even when the tax revenue is redistributed, and the entrepreneur-optimal tax rate is slightly above 40\%.

Although our proposed fixed-$k$ tail regression is motivated by the Forbes 400 data set, it is not limited to this particular application. 
From a theoretical perspective, our method allows for the truncation from both the top/right and the bottom/left as long as more than three extreme order statistics are observed \citep[cf.][]{WangXiao2021}. 
From an empirical perspective, our method can be extended to study other sampling processes when panel or repeated cross-sectional observations of $Y_{it}$ are available. 
Examples of such data sets can be found in the literature about firm sizes \citep[e.g.,][]{Axtell2001}, extreme infant birthweights \citep[e.g.,][]{Abrevaya2001}, medical insurance claims \citep[e.g.,][]{Manning2005}, and earnings risks \citep[e.g.,][]{Guvenen2021}, among others. 
All these data sets involve heavy tails in the distribution of $Y$, which are potentially affected by some covariate $X$ such as policy changes. 
Moreover, these data sets are also potentially affected by bottom truncation/censoring. 
For example, the income data is bottom truncated due to the minimum wage and the medical insurance claims are censored at zero. 

We close this paper with a list of limitations and directions for future research.
First, while \citet{Feldstein1995} describes the Tax Reform Act of 1986 leading to a decline in the tax rates as a ``natural experiment,'' there may be other latent factors on the wealth accumulation process that may have simultaneously changed during the 1980s.
If this is the case, our analysis that exploits the changes in the tax rates during this period may be potentially subject to omitted variables biases.
To at least partially account for this issue, we ran additional tail regressions including macroeconomic indicators to control for some latent factors and obtain similar estimates to our baseline estimates -- see Appendix \ref{sec:additional}.
That being said, one direction of future research to overcome this issue is to develop a new method of tail regression that can exploit exogenous variations (e.g., of an instrument) and to find such an instrument in data.

Second, given the nature of the event analysis with the tax decline perceived as a single-time event, difference-in-difference type analysis may be useful as well.
This will require a control group (e.g., countries with no tax decline in the 1980s) as well as a treatment group (e.g., United States).
But this approach introduces another difficulty, such as the requirement of a common trend assumption for identification, which is hardly plausible in cross-country analyses.
An additional difficulty with this approach is that extreme wealth data are not available for many other countries than the United States.
In this light, another direction of future research is to develop an econometric method for event study with more plausible assumptions (than the standard common trend assumptions) and to develop an extreme wealth value database across countries.

Third, while we focus on the common income tax rates, this aspect of the tax information only incompletely explains the wealth accumulation process. For instance, we do not use information about heterogeneous tax rates across states, as the Forbes 400 data do not contain information about the identity of all the listed individuals and thus their geographical attributes.
We also do not account for the possibility of retaining earnings through
the company and not having to pay taxes every period \citep*{acemoglu2020does}.
To at least partially account for the this possibility, we consider including corporate tax rates in the additional regressions run in Appendix \ref{sec:additional}.
That said, an important direction for future research is to acquire more detailed tax information about the entrepreneurs listed in the extreme order database.

\vspace{1cm}
\appendix
%%%%%%%%%%%%%%%%%%%%%%%%%%%%%%%%%%%%%%%%%%%%%%%%%%%%%%%%%%%%%%%%%%%%%%
\section*{Appendix}
%%%%%%%%%%%%%%%%%%%%%%%%%%%%%%%%%%%%%%%%%%%%%%%%%%%%%%%%%%%%%%%%%%%%%%
Appendix \ref{sec:proof} contains mathematical proofs. 
Appendix \ref{sec:robust} contains an additional theoretical result about robustness to measurement error. 
Appendix \ref{sec:additional} contains additional empirical estimation results. 
Appendix \ref{sec:simulation} contains simulation studies.

%%%%%%%%%%%%%%%%%%%%%%%%%%%%%%%%%%%%%%%%%%%%%%%%%%%%%%%%%%%%%%%%%%%%%%
\section{Mathematical proofs}\label{sec:proof}
%%%%%%%%%%%%%%%%%%%%%%%%%%%%%%%%%%%%%%%%%%%%%%%%%%%%%%%%%%%%%%%%%%%%%%

\subsection{Proof of Section \ref{subsec:theory} results}\label{subsec:proofmain}

We first establish a lemma about the Hellinger distance between generic
densities $f$ and $g$. 
Define the Hellinger distance between $f$ and $g$ by  
\begin{equation*}
H(f,g) =\left( \int \left( f(t)^{1/2}-g^{1/2}\left( t\right)
\right)^2dt\right) ^{1/2}.
\end{equation*}

\begin{lemma}\label{lemma:hellinger}
Conditional on $X_t=x$ for any $x$, the Hellinger
distance between $f_{\bY_t^*|X_t=x}$ and $f_{\bV_t^*|\xi_t=\Lambda \left( x^\intercal \theta_0\right) }$
satisfies 
\begin{equation}
H\left( f_{\bY_t^*|X_t=x},f_{\bV_t^*|\xi
_t=\Lambda \left( x^\intercal \theta_0\right) }\right) =O\left(
n^{-\beta(x) /\alpha (x)}\right) +O\left( 1/n\right)
\label{lemma1}
\end{equation}%
for any fixed $k$ as $n\to \infty$.
\end{lemma}

\begin{proof}[Proof of Lemma \protect\ref{lemma:hellinger}]
Since this proof is conditional on $X_t=x$, we therefore omit the
subscript $t$ for notational simplicity. Also denote $\xi_0=\Lambda
\left( x^\intercal \theta_0\right) $.

Let $f_{a_n^{-1}\left( \bY-b_n\right) |X=x}$ denote the density of 
$a_n^{-1}\left( \bY-b_n\right) $ conditional on $X=x$, and
similarly let $f_{\bV|\xi =\Lambda \left( x^\intercal \theta_0\right) }$
denote the density of $\bV$ conditional on $X=x$. Following \citet{Smith1987approx},
we define the following objects, all of which are conditional on $X=x$:
\begingroup
\allowdisplaybreaks 
\begin{align*}
\phi (v) &=\frac{1-F_{Y|X=x}(v) }{f_{Y|X=x}(v) }, \hspace{2cm}
h_{\beta }(v) =-\frac{v^{-\beta(x) }-1}{\beta
(x) }, \\
H_{\beta }(v,\eta) &=\frac{h_{\beta }(1+v\eta)
-\beta(x) h_{1}(1+v\eta) -(1-\beta(x)) \log (1+v\eta) }{-\beta(x)(1-\beta(x)) \eta ^{3}}, \\
b_n &=Q_{Y|X=x}(1-1/n), \\
\xi_n &=\partial \phi (v) /\partial v|_{v=b_n}, \\
r_n &=n^{-\beta(x) /\alpha (x)}, \\
\psi_0(v) &=(1+\xi v) ^{-1/\xi }, \\
\psi_n(v;c) &=(1+\xi_n v) ^{-1/\xi_n}\left(
1+cH_{\beta }(v,\xi_n) \right), \\
g_n(\by) &=\prod_{i=1}^{k}( -\psi_n'( y_i;cr_n)) \exp( -\psi_n(y_k;cr_n))
\end{align*}
\endgroup
for $y_{1}\ge \dots \ge y_k$ and $1+y_i\xi_n>0$ for each $i=1,\dots,k$. Let $\psi_n'$ denote the derivative of $\psi_n$ with respect to its first argument. Recall
the notation $\by=(y_{1},\dots,y_k)$.

Using Theorem 3.6 of \cite{Smith1987approx}, we have 
\begin{equation}
H\left( f_{a_n^{-1}(\bY_t-b_n)
|X_t=x},g_n\right) =o\left( n^{-\beta(x) /\alpha
(x)}\right) +O(1/n).  \label{penultimate}
\end{equation}%
Thus, the proof is complete by the following two steps. First, we show that 
\begin{equation}
H\left( g_n,f_{\bV|\xi =\Lambda (x^\intercal \theta
_0) }\right) =O\left( n^{-\beta(x) /\alpha (x)}\right) ,
\label{intermediate}
\end{equation}
which in turn yields  
\begin{equation}
H\left( f_{a_n^{-1}(\bY-b_n) |X_t=x},f_{\bV|
\xi =\Lambda( x^\intercal \theta_0) }\right) =O\left(
n^{-\beta(x) /\alpha (x)}\right) +O( 1/n)
\label{ultimate}
\end{equation}
by the triangle inequality. Second, we show that the densities of the
self-normalized statistics $\bY^*$ and $\bV^*$
also share the same convergence rate as desired in \eqref{lemma1}.

To establish \eqref{intermediate}, note that, conditional on $X=x$, 
\begingroup
\allowdisplaybreaks
\begin{align}
& \sup_{v\in [s_1,s_2] }\abs{\frac{\psi
_n'(v;cr_n) }{\psi_0'(v) }}\notag \\
& \le \sup_{v\in [s_1,s_2] }\abs{\frac{\left( 1+\xi
_nv\right) ^{-1/\xi_n-1}\left( 1+cr_nH_{\beta }\left( v,\xi
_n\right) \right) }{(1+\xi v) ^{-1/\xi -1}}}
+\sup_{v\in [s_1,s_2] }r_n\abs{\frac{c\left( 1+\xi
_nv\right) ^{-1/\xi_n}\frac{\partial H_{\beta }(v,\xi_n) 
}{\partial v}}{(1+\xi v) ^{-1/\xi -1}}} \notag \\
& \le 1+O\left( \xi_n-\xi \right) +O\left( r_n\right) \label{ineq1}\\
& = 1+O\left( r_n\right) \label{ineq2}
\end{align}
\endgroup
for any
compact set $[s_1,s_2] $ that is strictly within $(\max
\{-1/\xi_n,-1/\xi_0\},\infty )$,
where \eqref{ineq1} follows the mean value expansion and \eqref{ineq2} follows from the fact
that $O\left( \xi_n-\xi \right) =O\left( b_n^{-\beta(x)
}\right) =O\left( n^{-\beta(x) /\alpha (x) }\right)
=O\left( r_n\right) $; see, for example, \citet[p.~7]{Smith1987approx}.

Similarly, we have
\begingroup
\allowdisplaybreaks
\begin{align*}
&\sup_{v\in [s_1,s_2] }\exp \left( -\psi_n\left(
v;cr_n\right) +\psi_0(v) \right) \\
&=\sup_{v\in [s_1,s_2] }\exp \left( -\left( 1+\xi
_nv\right) ^{-1/\xi_n}\left( 1+cr_nH_{\beta }(v,\xi_n)
\right) +(1+\xi v) ^{-1/\xi }\right) \\
&=\exp \left( O\left( r_n\right) \right)
=1+O\left( r_n\right) .
\end{align*}
\endgroup
Therefore,
\begingroup
\allowdisplaybreaks
\begin{align}
&\int_{s_1}^{s_2}\abs{ g_n\left( \by\right) -f_{\bV|\xi =\Lambda( x^\intercal \theta_0) }(\by)
} \diff \by  \notag \\
&=\int_{s_1}^{s_2}\abs{ \prod_{i=1}^{k}\frac{\psi_n'( y_i;cr_n) }{\psi_0'(y_i) }\exp
\left( -\psi_n( y_i;cr_n) +\psi_0( y_i)
\right) -1} f_{\bV|\xi =\Lambda( x^\intercal \theta
_0) }(\by) \diff \by  \label{bound} \\
&=O(r_n). \notag
\end{align}
\endgroup
Once we extend the bound \eqref{bound} to allow $1+\xi s_1\to 0$ and $s_2\to \infty$, it is then sufficient for \eqref{intermediate}
since the bound in total variation implies the bound in Hellinger distance.
The extension with $s_2\to \infty$ is straightforward since $f_{\bV|\xi =\Lambda( x^\intercal \theta_0) }(\by) $ decays exponentially as $1+\xi y_k\to \infty$. Now
consider the case with $1+\xi s_1\to 0$. We have 
\begingroup
\allowdisplaybreaks
\begin{align*}
\frac{\psi_n'(v;cr_n) }{\psi_0'(v) }-1 &\asymp r_n\log (1+\xi v) \quad \text{and} \\
\exp \left( -\psi_n(v;cr_n) +\psi_0(v)
\right) -1 &\asymp r_n(1+\xi v) ^{-1/\xi -1}
\end{align*}%
\endgroup
as $1+\xi v\to 0$,
where we use $\asymp $ to denote equivalence in the order of magnitude. Then,
the facts (e.g., \citet[p.~17]{Smith1987approx}) that
\begin{align*}
&\int_{-1/\xi }^{\infty }\log (1+\xi v) \psi_0'(v) \exp (-\psi_0(v) )\diff v < \infty \quad \text{and} \\
&\int_{-1/\xi }^{\infty }(1+\xi v) ^{-1/\xi -1}\psi_0'(v) \exp (-\psi_0(v) )\diff v < \infty
\end{align*}%
yield the bound in \eqref{bound} for $1+\xi s_1\to 0$.

To establish \eqref{lemma1}, consider the change of variables $\by\to (\by^*,z,y_k)$ defined by $z = y_1-y_k$ and $y_i^* =\frac{y_i-y_k}{z}$ for $i=2,\dots,k-1$. Then
\begin{align*}
f_{\bY^*|\xi =\Lambda(x^\intercal \theta
_0) }( \by^*) &=\iint_0^{\infty
}z^{k-2}f_{a_n^{-1}(\bY-b_n) |X=x}\left(
y_k+z,y_{2}^*z+y_k,\dots,y_{k-1}^*z+y_k,y_k\right)
\diff z \diff y_k,\\
f_{\bV^*|\xi =\Lambda( x^\intercal \theta
_0) }( \by^*) &=\iint_0^{\infty
}z^{k-2}f_{\bV|\xi =\Lambda( x^\intercal \theta_0)
}\left( y_k+z,y_{2}^*z+y_k,\dots,y_{k-1}^*z+y_k,y_k\right)
\diff z \diff y_k.
\end{align*}
We can now complete the proof using
\begin{align*}
&\int \abs{f_{\bY^*|\xi =\Lambda(
x^\intercal \theta_0) }(\by^*) -f_{\bV^*|\xi =\Lambda( x^\intercal \theta_0) }( 
\by^*) } \diff \by^* \\
&\le \iiint_0^{\infty }z^{k-2}\Bigl\lvert 
f_{a_n^{-1}(\bY-b_n) |X_t=x}\left(
y_k+z,y_{2}^*z+y_k,\dots,y_{k-1}^*z+y_k,y_k\right) \\ 
&\hphantom{\le \iiint_0^{\infty }z^{k-2}\Bigl\lvert}-f_{\bV|\xi =\Lambda( x^\intercal \theta_0) }\left(
y_k+z,y_{2}^*z+y_k,\dots,y_{k-1}^*z+y_k,y_k\right)
\Bigr\rvert \diff z\diff y_k\diff \by^* \\
&=\int \abs{f_{a_n^{-1}(\bY-b_n)
|X_t=x}(\by) -f_{\bV|\xi =\Lambda(
x^\intercal \theta_0) }(\by)} \diff \by =O(r_n) +O(n^{-1}),
\end{align*}
where the last equality follows from \eqref{ultimate}.
\end{proof}

Using Lemma \ref{lemma:hellinger}, we are now ready to prove Theorem \ref{thm:asym}.

\begin{proof}[Proof of Theorem \ref{thm:asym}]

Recall that
\begin{equation*}
\hat{\theta}=\arg \max_{\theta \in \Theta }S_{T}\left( \theta \right) =-\frac{1}{T}\sum_{t=1}^{T}\log f_{\bV^*|\Lambda \left(X_{t}^{\intercal }\theta \right) }\left( \bY_{t}^*\right), 
\end{equation*}
where the density $f_{\bV^*|\Lambda(X_t^{\intercal}\theta)}$ is given in \eqref{fvstar}. 
Let $\underline{\xi}=\inf_{(x,\theta)\in \mathcal{X}\times \Theta} > 0 $ from Condition \ref{cond3}. 
By some straightforward but tedious calculation, we can verify that this density and its derivative $ \dot{f}_{\bV^{*}|\xi}(\by^{*}) =\frac{\partial }{\partial \xi}f_{\bV^{*}|\xi}(\by^{*}) $ satisfy the following properties:
\begin{align}
\sup_{\xi \in [\underline{\xi},1]} \int_{\mathcal{Y}} \left( \log f_{\bV^{*}|\xi }\left( \mathbf{y}^{*}\right) \right) ^{4} \diff \mathbf{y}^* < \infty, \label{fV_fact1}\\
\sup_{\xi \in [\underline{\xi},1]} \int_{\mathcal{Y}} \abs{\dot f_{\bV^{*}|\xi }\left( \mathbf{y}^{*}\right) } ^{8}  \diff \mathbf{y}^* < \infty, \label{fV_fact2}\\
\sup_{\xi \in [\underline{\xi},1]} \int_{\mathcal{Y}} \abs{\frac{\partial^2 \log f_{\bV^{*}|\xi}(\by^{*})    }{\partial \xi^2   } }^4 \diff \by^{*} < \infty, \label{fV_fact3}
%\sup_{\xi \in [\underline{\xi},1]} \int_{\mathcal{Y}} \abs{\frac{\partial^3 \log f_{\bV^{*}|\xi}(\by^{*})    }{\partial \xi^3   } }^2 \diff \by^{*} < \infty, \label{fV_fact4}
\end{align} 
where $ \mathcal{Y} = \{1\}\times \{(v^{*}_2,\dots,v^{*}_{k-1}) \in (0,1)^{k-2}:v^{*}_2 \geq \dots \geq v^{*}_{k-1} \} \times \{0\}$ denotes the support of $\bY_t^{*}$. 
Let $C$ denote a generic constant, whose value can change line by line. 

We start with showing the following consistency result.
\begin{equation}\label{consistency}
 \hat{\theta}-\theta _{0}\overset{p}{\rightarrow }0\text{ as }T\rightarrow \infty.
\end{equation}  
First, $S_{T}\left( \theta \right) $ is continuous in $\theta \in \Theta $ for all $X_{t}$ and $\bY_{t}^*$ since $f_{\bV^*|\xi }$ is continuous in $\xi $ and $\Lambda \left( \cdot\right) $ is continuous. 
By the standard argument \citep[e.g.,][Theorem 4.1.1]{Amemiya1985}, it suffices to show that
\begin{enumerate} 
\item\label{item:estab1} $\sup_{\theta \in \Theta}\left\vert S_{T}\left( \theta \right) -S\left( \theta \right) \right\vert \overset{p}{\rightarrow }0$ where $S\left( \theta \right) =\E\left[ \log f_{\bV^{*}|\Lambda \left( X^{\intercal }\theta \right) }\left( \bV^{*}\right) \right] $, and 
\item\label{item:estab2} $\theta _{0}$ uniquely maximizes $S\left( \theta \right) $. 
\end{enumerate} 

To show \ref{item:estab1}, we show the pointwise convergence and the stochastic equicontinuity. 
For the former, Condition 1 implies that for each $t$ conditional on $X_{t}=x$, the largest $k$ order statistics $\set{Y_{(1)t},\dots,Y_{(k)t}}$ are a random draw from the distribution $F_{Y|X=x}\left( y_{(k)t}\right) ^{n-k}\Pi _{i=1}^{k}\left( k-i+1\right) f_{Y|X=x}\left( y_{(i)t}\right) $ \citep[e.g.,][p.~219]{Arnold2008}. 
Therefore, $\{Y_{(1)t},\dots,Y_{(k)t},X_{t}\}$ is an $\alpha $-mixing triangular array process with the same mixing coefficients as $X_{t}$. 
By the pointwise weak law of large numbers for $\alpha $-mixing triangular array \citep[e.g.,][Theorem 2]{Andrews1988}, $S_{T}\left( \theta \right) -\E \left[ S_{T}\left(\theta \right) \right] \overset{p}{\rightarrow }0$ for each $\theta \in \Theta $. 
Then, the pointwise convergence follows from
\begin{align}
& \left\vert \E[ S_{T}\left( \theta \right) ] - S\left(\theta \right) \right\vert \notag \\
&= \left\vert -\E_{X_{t}}\left[ \E_{\bY_{t}^*}\left[ \left. \log f_{\bV^*|\Lambda \left( X_{t}^{\intercal}\theta \right) }\left( \bY_{t}^*\right) \right\vert X_{t}\right] \right] -S\left( \theta \right) \right\vert \notag \\
&=\left\vert \E_{X_{t}}\left[ \int \left. \log f_{\bV^{*}|\Lambda \left( X_{t}^{\intercal }\theta \right) }\left( \mathbf{y}^{*}\right) \left( f_{\bY_{t}^*|X_{t}}\left( \mathbf{y}^{*}\right) -f_{\bV^*|\Lambda \left( X_{t}^{\intercal }\theta
\right) }\right) \diff \by^{*} \right\vert X_{t}\right] \right\vert  \notag \\
&\leq \E_{X_{t}}\left[ \left( \int \left( \log f_{\bV^{*}|\Lambda \left( X_{t}^{\intercal }\theta \right) }\left( \mathbf{y}^{*}\right) \right) ^{2}d\mathbf{y}^*\right) ^{1/2}\left( \int \left( f_{%
\bY_{t}^*|X_{t}}\left( \mathbf{y}^*\right) -f_{\bV^*|\Lambda \left( X_{t}^{\intercal }\theta \right) }\right) ^{2} \diff \by^{*} \right) ^{1/2}\right]  \label{ESs1}\\
&\leq C\E_{X_{t}}\left[ n^{-\beta \left( X_{t}\right) /\alpha \left( X_{t}\right) }+n^{-1}\right]  \label{ESs2} \\
&\leq C ( \sup_{x}n^{-\beta \left( x\right) /\alpha \left( x\right) } + n^{-1}) = o(1), \label{ESs3}
\end{align}
where \eqref{ESs1} follows from Cauchy-Schwarz inequality, 
\eqref{ESs2} follows from \eqref{fV_fact1}, 
and \eqref{ESs3} is from Condition \ref{cond4}. 

For the stochastic equicontinuity, we resort to \citet[][Theorem 1]{PotscherPrucha1989ULLN}. 
In particular, their Assumption 3 (the dominance condition) is implied by \eqref{fV_fact1} and Lemma 1. 
Their remaining assumptions are satisfied by our Conditions \ref{cond1}-\ref{cond4}. 

For \ref{item:estab2}, consider $\tilde{S}\left( \theta \right) =\E \left[ \log f_{\bV^*|\Lambda \left( X_{t}^{\intercal }\theta \right) }\left( 
\bV^*\right) +\log f_{X}\left( X_{t}\right) \right] ,$ where $f_{X}\left( \cdot \right) $ denotes the density of $X_{t}$, which is
uniquely defined by the strict stationarity. Since $f_{X}$ does not involve $\theta $, it is equivalent to show that $\theta _{0}$ uniquely maximizes $%
\tilde{S}\left( \theta \right) $. Note that the joint density of $\left( \bV^*,X\right) $ is $f_{\bV^*|\Lambda \left( X^{\intercal }\theta \right) }\left( \cdot \right) f_{X}\left( \cdot \right) $. 
Thus, it suffices to show that the Fisher information matrix is invertible, which follows from Condition \ref{cond4}. 
The consistency of $\hat{\theta}$ then follows. 

Now, we are going to show the asymptotic normality
\begin{equation}
T^{1/2}\left( \hat{\theta}-\theta _{0}\right) \Rightarrow \mathcal{N}\left(0, \cI(\theta_0)^{-1} \mathcal{W} (\theta_0) \cI(\theta_0)^{-1} \right) \text{ as }T\rightarrow \infty ,
\label{eq:asym_normality}
\end{equation}%
where $\mathcal{W} (\theta_0) $ denotes the long-run variance
\begin{equation*}
\mathcal{W} (\theta_0) =\lim_{T\rightarrow \infty }\Var\left[ \frac{%
1}{T^{1/2}}\sum_{t=1}^{T}\frac{\partial \log f_{\bV^*|\Lambda
\left( X^{\intercal }\theta \right) }\left( \bY_{t}^*\right) }{%
\partial \theta }\right] .
\end{equation*}
Using the mean-value expansion and the fact that $f_{\bV^{*}|\Lambda
(\cdot)}$ is continuously differentiable, we have
\begin{align*}
\hat{\theta}-\theta_0&=-\left( T^{-1}\sum_{t=1}^{T}\frac{\partial^2\log f_{\bV^{*}|\Lambda (X_t^\intercal \dot{\theta})}(\bY_t^{*})}{\partial
\theta \partial \theta^\intercal }\right) ^{-1}\times \left(
T^{-1}\sum_{t=1}^{T}\frac{\partial \log f_{\bV^{*}|\Lambda
(X_t^\intercal \theta_0)}(\bY_t^{*})}{\partial \theta }
\right) \\
&\eqqcolon -Q_{T1}^{-1}\times Q_{T2}
\end{align*}%
for some value $\dot{\theta}$ that lies on the line segment connecting $\theta_0$ and $\hat{\theta}$.
It remains to establish
\begin{enumerate*}\addtocounter{enumi}{2}
\item\label{item:estab3} $Q_{T1}\pto \cI(\theta_0)$ and
\item\label{item:estab4} $T^{1/2}Q_{T2}\Rightarrow 
\cN(0,\mathcal{W} (\theta_0))$.
\end{enumerate*}

For part \ref{item:estab3}, given the consistency \eqref{consistency}, it suffices to show the uniform convergence that for some shrinking open ball $\mathcal{B}_{T} (\theta_0) \subset \mathbb{R}^{p}$ centered around $\theta_0$, 
\begin{equation*}
\sup_{\theta \in \mathcal{B}_{T} (\theta_0) }\abs{ T^{-1}\sum_{t=1}^{T}\frac{\partial^2\log f_{\bV^{*}|\Lambda (X_t^\intercal \theta )}(\bY_t^{*})}{\partial \theta^2}-\E\left[ \frac{\partial^2\log f_{\bV^{*}|\Lambda (X_t^\intercal \theta )}(\bY_t^{*})}{\partial \theta^2}\right]} =o_{p}(1). 
\end{equation*} 
This follows from \citet[][Theorem 1]{PotscherPrucha1989ULLN} with the dominance condition implied by \eqref{fV_fact3}.  

Second, we have
\begin{align*}
& \abs{-\E\left[ \frac{\partial^2\log f_{\bV^{*}|\Lambda (X_t^\intercal \theta)}(\bY_t^{*})}{\partial \theta^2}\right] -\cI(\theta)} \\
& =\abs{-\E_{X_t}\left[ \E_{\bY_t^{*}}\left[ \frac{\partial^2\log f_{\bV^{*}|\Lambda
(X_t^\intercal \theta)}(\bY_t^{*})}{\partial \theta^2} \middle| X_t\right] \right] -\cI(\theta)} \\
& \le \abs{ \E_{X_t}\left[ \int \frac{\partial^2\log f_{%
\bV^{*}|\Lambda (X_t^\intercal \theta)}(\by^{*})}{\partial \theta^2}\left( f_{\bY_t^{*}|X_t}(\by^{*}) -f_{\bV^{*}|\Lambda (X_t^{\intercal} \theta) }(\by^{*})\right) \diff \by^{*}\right]} \\
& \hphantom{=}+\underbrace{\abs{-\E_{X_t}\left[ \int \frac{\partial^2\log f_{\bV^{*}|\Lambda (X_t^\intercal \theta)}(\by^{*})}{\partial \theta^2}f_{\bV^{*}| \Lambda( X_t^{\intercal}\theta) }(\by^{*}) \diff \by^{*}\right] - \cI(\theta)}}_{=0} \\
& \le C\E_{X_t}\left[ \int \abs{f_{\bY_t^{*}|X_t}(\by^{*}) -f_{\bV^{*}|\Lambda( X_t^{\intercal} \theta) }(\by^{*})} \diff \by^{*}\right]= O(r_n) =o(1),
\end{align*}
where the last equality follows from \eqref{fV_fact3} and the proof of Lemma 1. 
Now, \ref{item:estab3} follows by combining the above two steps and the fact that $\mathcal{I}(\theta)$ is continuous in $\theta$. 

For part \ref{item:estab4}, we have
\begin{align}
& \left\vert \E\left[ Q_{T2}\right] \right\vert \notag \\
&=\left\vert \E_{X_{t}}\left[ \E_{\bY_{t}^{* }}\left[ \left. \frac{\dot{f}_{\bV^{* }|\Lambda \left(X_{t}^{\intercal }\theta _{0}\right) }\left( \bY_{t}^{* }\right) }{f_{\bV^{* }|\Lambda \left( X_{t}^{\intercal }\theta _{0}\right)}\left( \bY_{t}^{* }\right) }\Lambda ^{\prime }\left(X_{t}^{\intercal }\theta _{0}\right) X_{t}\right\vert X_{t}\right] \right]\right\vert \notag \\
&\leq \left\vert \E_{X_{t}}\left[ \Lambda ^{\prime }\left(X_{t}^{\intercal }\theta _{0}\right) X_{t}\left( \int \frac{\dot{f}_{\bV^{* }|\Lambda \left( X_{t}^{\intercal }\theta _{0}\right) }\left( \by^{* }\right) }{f_{\bV^{* }|\Lambda \left(
X_{t}^{\intercal }\theta _{0}\right) }\left( \by^{* }\right) }\left( f_{\bY_{t}^{* }|X_{t}}\left( \by^{* }\right) -f_{\bV^{* }|\Lambda \left( X_{t}^{\intercal }\theta _{0}\right)}\left( \by^{* }\right) \right) \diff \by^{* }\right) \right]\right\vert \notag \\
& +\underset{=0}{\underbrace{\left\vert \E_{X_{t}}\left[ \Lambda^{\prime }\left( X_{t}^{\intercal }\theta _{0}\right) X_{t}\int \dot{f}_{\bV^{* }|\Lambda \left( X_{t}^{\intercal }\theta _{0}\right) }\left( \by^{* }\right) \diff \by^{* }\right] \right\vert }} \label{QTs1} \\
&\leq \sup_{\xi \in \left[ 0,1\right] }\left( \int \left( \frac{\dot{f}_{\bV^{* }|\xi}\left( \by^{* }\right) }{f_{\bV^{* }|\xi }\left( \by^{* }\right) }\right) ^{2} \diff \by^{* }\right) ^{1/2}\left(\E_{X_{t}}\left[ \Lambda ^{\prime }\left( X_{t}^{\intercal }\theta _{0}\right) ^{2}\left\vert \left\vert X_{t}X_{t}^{\intercal }\right\vert \right\vert \right] \right)^{1/2} \label{QTs2} \\
&\times \left( \E_{X_{t}}\left[ \int \left( f_{\bY_{t}^{*}|X_{t}}\left( \by^{* }\right) -f_{\bV^{* }|\Lambda \left( X_{t}^{\intercal }\theta _{0}\right) }\left( \by^{* }\right) \right) ^{2} \diff \by^{* }\right] \right) ^{1/2} \label{QTs3} \\
&\leq C\left( \sup_{x}n^{-\beta \left( x\right) /\alpha \left( x\right) }\right) ^{1/2}, \label{QTs4}
\end{align}
where \eqref{QTs1} follows from the fact that $\int \dot{f}_{\bV^{*}|\Lambda (X_t^\intercal \theta_0)}(\by^{*}) \diff \by^{*}=0$, as implied by Leibniz rule; 
\eqref{QTs2} and \eqref{QTs3} follow from repeated applications of Cauchy-Schwarz inequality; 
and \eqref{QTs4} is due to \eqref{fV_fact3} with the fact that $f_{\bV^*|\xi}(\by^*)$ is uniformly bounded below from zero over $\xi \in [0,1]$ and $\by^* \in \mathcal{Y}$, Condition \ref{cond3}, and Lemma 1. 

The above derivation and Condition \ref{cond4} imply
\begin{equation*}
T^{1/2}\E\left[ Q_{T2}\right] \le
T^{1/2}C \left( \sup_{x}n^{-\beta(x)/\alpha(x)} \right)^{1/2}=o(1).
\end{equation*}
Furthermore, $\Var\left[ T^{1/2}Q_{T2}\right] \to \mathcal{W} (\theta_0)$, which is well-defined given Condition \ref{cond1}. 
Thus, the proof follows from the CLT with stationary and $\alpha$-mixing triangular array \citep[e.g.,][Theorem 1]{Ekstrom2014} given Condition \ref{cond1} and \eqref{fV_fact2} and Slutsky's theorem.
\end{proof}

\subsection{Proof of Section \ref{sec:model} results}\label{subsec:proofmodel}

We solve the optimal decision problem of an entrepreneur when the elasticity of intertemporal substitution takes an arbitrary value $\varepsilon>0$. In this case the aggregator is
\begin{equation}
h(c,v)=\beta [(1-\gamma)v]^\frac{1/\varepsilon-\gamma}{1-\gamma}\frac{c^{1-1/\varepsilon}-[(1-\gamma)v]^\frac{1-1/\varepsilon}{1-\gamma}}{1-1/\varepsilon}.\label{eq:hcv}
\end{equation}
The expression \eqref{eq:hcv1} is a special case by taking the limit of \eqref{eq:hcv} as $\varepsilon \to 1$.

The following proposition generalizes Proposition \ref{prop:agent1}.

\begin{proposition}\label{prop:agent}
Consider an entrepreneur facing expected return $\mu_t=f'(y_t)$, tax rate $\tau_t$, and volatility $\sigma$. Suppose that a solution to the individual optimization problem exists and let $J(t,w)$ be the value function (the maximized utility in \eqref{eq:utility}). Then the solution is given by
\begin{subequations}
\begin{align}
J(t,w)&=\frac{1}{1-\gamma}a_t^{1-\gamma}w^{1-\gamma},\label{eq:optJ}\\
\theta_t&=\frac{(\mu_t-r_t)_+}{(1-\tau_t)\gamma\sigma^2},\label{eq:optTheta}\\
m_t&=\beta^\varepsilon a_t^{1-\varepsilon},\label{eq:optC}
\end{align}
\end{subequations}
where $x_+=\max\set{x,0}$ and the coefficient of the value function $a_t>0$ satisfies the ordinary differential equation (ODE)
\begin{equation}
\frac{\dot{a}_t}{a_t}=-\left((1-\tau_t)r_t+\frac{(\mu_t-r_t)_+^2}{2\gamma\sigma^2}\right)+\beta
\begin{cases}
\frac{1}{\varepsilon-1}(\varepsilon-\beta^{\varepsilon-1} a_t^{1-\varepsilon}), & (\varepsilon\neq 1)\\
(1-\log\beta+\log a_t).& (\varepsilon=1)
\end{cases}\label{eq:aODEinf}
\end{equation}
In steady state ($r_t=r$, $\mu_t=\mu$, and $\tau_t=\tau$ are constant), the optimal portfolio, consumption, and the coefficient of the value function are
\begin{subequations}\label{eq:steady}
\begin{align}
\theta&=\frac{(\mu-r)_+}{(1-\tau)\gamma\sigma^2},\label{eq:steadytheta}\\
m&=\varepsilon\beta+(1-\varepsilon)\left((1-\tau)r+\frac{(\mu-r)_+^2}{2\gamma\sigma^2}\right),\label{eq:steadyc}\\
a&=\begin{cases}
\beta^\frac{1}{1-1/\varepsilon}\left[\varepsilon\beta+(1-\varepsilon)\left((1-\tau)r+\frac{(\mu-r)_+^2}{2\gamma\sigma^2}\right)\right]^\frac{1}{1-\varepsilon}, &(\varepsilon\neq 1)\\
\beta\exp\left(\frac{1}{\beta}\left((1-\tau)r+\frac{(\mu-r)_+^2}{2\gamma\sigma^2}\right)-1\right). & (\varepsilon=1)
\end{cases}\label{eq:steadya}
\end{align}
\end{subequations}
\end{proposition}

\begin{proof}%[Proof of Proposition \ref{prop:agent}]
We can solve the individual optimization problem using stochastic control \citep{FlemingSoner2006}. For notational simplicity, suppress the time subscript $t$. Using the budget constraint \eqref{eq:bc2}, the Hamilton-Jacobi-Bellman (HJB) equation becomes
\begin{equation}
0=\max_{m,\theta\ge 0}\left[h(mw,J)+J_t+J_w(\tilde{r}+\tilde{\alpha}\theta-m)w+\frac{1}{2}J_{ww}(\tilde{\sigma}\theta w)^2\right],\label{eq:HJB}
\end{equation}
where $h$ is given by \eqref{eq:hcv}, $\tilde{r}=(1-\tau)r$, $\tilde{\alpha}=(1-\tau)(\mu-r)$, and $\tilde{\sigma}=(1-\tau)\sigma$.

Assume $\varepsilon\neq 1$. Due to homogeneity, the value function takes the form $J(t,w)=\frac{1}{1-\gamma}a_t^{1-\gamma}w^{1-\gamma}$, where $a_t>0$. Fixing the state variables $t,w$, let
\begin{equation*}
H(m,\theta)=h(mw,J)+J_t+J_w(\tilde{r}+\tilde{\alpha}\theta-m)w+\frac{1}{2}J_{ww}(\tilde{\sigma}\theta w)^2
\end{equation*}
be the expression inside the bracket of \eqref{eq:HJB}. Since $J_{ww}=-\gamma a_t^{1-\gamma}w^{-\gamma-1}<0$, $F$ is a concave quadratic function of $\theta$. Therefore it takes the global maximum over $\theta\in \R$ when
\begin{equation*}
0=\frac{\partial H}{\partial \theta}=J_w\tilde{\alpha}w+J_{ww}\tilde{\sigma}^2w^2\theta\iff \theta=\frac{\tilde{\alpha}}{2\gamma\tilde{\sigma}^2}=\frac{\mu-r_t}{(1-\tau)\gamma\sigma^2}.
\end{equation*}
Since $\theta\ge 0$, \eqref{eq:optTheta} is the solution.

Since by \eqref{eq:hcv1} $h(c,v)$ is concave in $c$, $H$ is concave in $m$. Therefore it takes the maximum over $m\ge 0$ when
\begin{equation*}
0=\frac{\partial H}{\partial m}=\beta [(1-\gamma)J]^\frac{1/\varepsilon-\gamma}{1-\gamma}(mw)^{-1/\varepsilon}w-J_ww\iff m=\beta^\varepsilon a_t^{1-\varepsilon} w,
\end{equation*}
which is \eqref{eq:optC}.

To derive an ODE for $a$, we use the HJB equation \eqref{eq:HJB}. Since $J$ is proportional to $w^{1-\gamma}$, $H$ is homogeneous of degree $(1-\gamma)$ in $w$. Therefore without loss of generality we may set $w=1$. In this case, using the first-order condition with respect to $m$, $[(1-\gamma)J]^\frac{1}{1-\gamma}=a$, $J_w=a^{1-\gamma}$, and $J_{ww}=-\gamma a^{1-\gamma}$, we can simplify \eqref{eq:HJB} as
\begin{equation}
0=\frac{1}{1-1/\varepsilon}J_wm-\frac{\beta}{1-1/\varepsilon}a^{1-\gamma}+a^{-\gamma}\dot{a}+\left(\tilde{r}+\tilde{\alpha}\theta-\frac{1}{2}\gamma\tilde{\sigma}^2\theta^2\right)a^{1-\gamma}-J_wm.\label{eq:temp.1}
\end{equation}
Since $J_wm=a^{1-\gamma}\beta^{\varepsilon}a^{1-\varepsilon}=\beta^\varepsilon a^{2-\varepsilon-\gamma}$, multiplying both sides of \eqref{eq:temp.1} by $(1-\varepsilon)a^\gamma$, we obtain
\begin{equation*}
0=-\beta^\varepsilon a^{2-\varepsilon}+\left[\varepsilon\beta+(1-\varepsilon)\left(\tilde{r}+\tilde{\alpha}\theta-\frac{1}{2}\gamma\tilde{\sigma}^2\theta^2\right)\right]a+(1-\varepsilon)\dot{a},
\end{equation*}
which is equivalent to \eqref{eq:aODEinf} if $\varepsilon\neq 1$. If $\varepsilon=1$, letting $\varepsilon\to 1$ in the formula \eqref{eq:aODEinf} with $\varepsilon\neq 1$ and using l'H\^opital's rule, we obtain the expression for $\varepsilon=1$.
\end{proof}

\begin{proof}[Proof of Proposition \ref{prop:welfare}]
Let $y=K/L$ be the capital-labor ratio. Then \eqref{eq:WW0} immediately follows from $\cW_W^0=\omega$ and \eqref{eq:lfoc}. \eqref{eq:T} follows from aggregating the budget constraint \eqref{eq:bc2} and using \eqref{eq:steadyK}. Using \eqref{eq:WW0} and \eqref{eq:T}, we obtain
\begin{equation*}
\cW_W^1=f(y)-yf'(y)+(f'(y)-\beta)y=f(y)-\beta y,
\end{equation*}
which is \eqref{eq:WW1}.

Since by Assumption \ref{asmp:prod} we have $f''<0$, by \eqref{eq:steadyK} the capital-labor ratio $y=K/L$ is strictly decreasing in $\tau$. Then
\begin{equation*}
\frac{\partial \cW_W^0}{\partial y}=(f(y)-yf'(y))'=-yf''(y)>0,
\end{equation*}
so $\cW_W^0$ is strictly decreasing in $\tau$. Furthermore, using \eqref{eq:steadyK}, we obtain
\begin{equation*}
\frac{\partial \cW_W^1}{\partial y}=f'(y)-\beta=\frac{\beta}{1-\tau}-\beta=\frac{\beta \tau}{1-\tau}>0,
\end{equation*}
so $\cW_W^1$ is also strictly decreasing in $\tau$.

To show \eqref{eq:WE}, note that by \eqref{eq:steadyK} the expected return satisfies
\begin{equation}
\mu=f'(y)=f'(K/L)=\frac{\beta}{1-\tau}.\label{eq:steadymu}
\end{equation}
Combining \eqref{eq:optTheta1}, \eqref{eq:bclear}, and \eqref{eq:steadymu}, the equilibrium interest rate must satisfy
\begin{equation}
r=\frac{\beta}{1-\tau}-(1-\tau)\gamma\sigma^2. \label{eq:steadyr}
\end{equation}
Combining $\cW_E=aK$, \eqref{eq:steadya}, \eqref{eq:optTheta1}, and \eqref{eq:steadyr}, it follows that
\begin{equation*}
\cW_E=\beta \exp\left(\frac{1}{\beta}\left(\beta-(1-\tau)^2\gamma\sigma^2+\frac{1}{2}(1-\tau)^2\gamma\sigma^2\right)-1\right)K=\beta \e^{-\frac{1}{2\beta}(1-\tau)^2\gamma\sigma^2}K,
\end{equation*}
which is \eqref{eq:WE}.

Now suppose the production function is Cobb-Douglas as in \eqref{eq:CD}. Then \eqref{eq:steadyK} implies
\begin{equation*}
\frac{\beta}{1-\tau}=f'(K/L)=A\rho (K/L)^{\rho-1}-\delta\iff K(\tau)=L\left[\frac{1}{A\rho}\left(\frac{\beta}{1-\tau}+\delta\right)\right]^\frac{1}{\rho-1}.
\end{equation*}
Therefore the logarithm of entrepreneur welfare becomes
\begin{equation*}
\log \cW_E=-\frac{1}{2\beta}(1-\tau)^2\gamma\sigma^2+\frac{1}{\rho-1}\log\left(\frac{\beta}{1-\tau}+\delta\right)+\text{constant}.
\end{equation*}
Letting $x=\frac{1}{1-\tau}$ and ignoring the constant term, the above expression becomes
\begin{equation*}
g(x)\coloneqq -\frac{\gamma\sigma^2}{2\beta x^2}+\frac{1}{\rho-1}\log (\beta x+\delta).
\end{equation*}
Taking the derivative with respect to $x$, we obtain
\begin{equation*}
g'(x)=\frac{\gamma\sigma^2}{\beta x^3}+\frac{1}{\rho-1}\frac{\beta}{\beta x+\delta}=\frac{1}{x}\left(\frac{\gamma\sigma^2}{\beta x^2} + \frac{1}{\rho-1}\frac{\beta x}{\beta x+\delta} \right)\eqqcolon \frac{h(x)}{x}.
\end{equation*}
Since $\rho\in (0,1)$, clearly $h$ is strictly decreasing in $x$. Furthermore, $h(0)=\infty$ and $h(\infty)=\frac{1}{\rho-1}<0$. Therefore $g$ is initially increasing and then decreasing in $x$, and because $x=\frac{1}{1-\tau}$ is monotonic in $\tau$, the entrepreneur welfare $\cW_E$ is single-peaked in $\tau$. Using \eqref{eq:CD} and \eqref{eq:T}, the tax revenue becomes $\cT=A\rho K^\rho L^{1-\rho}-\delta K$, which is strictly concave (single-peaked) in $K$. Because $K(\tau)$ in \eqref{eq:Ktau} is strictly decreasing in $\tau$, it follows that $\cT$ is single-peaked in $\tau$.
\end{proof}

%\singlespacing
%\bibliographystyle{econometrica}
%\bibliographystyle{plainnat}
%\bibliography{bib}

%\newpage
%\setcounter{page}{1}
%\onehalfspacing
%\begin{center}
%{\Large Supplementary Material for ``Tax Progressivity and Wealth Inequality:
%Evidence from Forbes 400''}
%\\
%\vspace{0.5in}
%Ji Hyung Lee\qquad
%Yuya Sasaki\qquad
%Alexis Akira Toda\qquad
%Yulong Wang
%\\
%\vspace{0.5in}
%\end{center}

%\begin{abstract}%\setlength{\baselineskip}{5.5mm}
%This supplementary material contains an additional theoretical result, additional empirical estimation results, and simulation studies.
%\end{abstract}

%%%%%%%%%%%%%%%%%%%%%%%%%%%%%%%%%%%%%%%%%%%%%%%%%%%%%%%%%%
\section{Robustness to measurement error}\label{sec:robust}
%%%%%%%%%%%%%%%%%%%%%%%%%%%%%%%%%%%%%%%%%%%%%%%%%%%%%%%%%%%%%%%%%%%%%%

In this section, we show that our asymptotic theory is robust to measurement error, which may arise in the Forbes 400 data set.
In particular, we aim to show that the self-normalized statistic $\bY_t^*$ has the same limit as before even when some random error is added.
We suppress the subscript $t$ and $X_t$ without loss of generality since such robustness is established to hold conditional on $X_t$ for any $t$.

To fix the idea, let $\varepsilon_i$ be some random measurement error, which is assumed to be independent from the true wealth $Y_i$ and also independent across $i$, though not necessarily identically distributed across $i$.
The econometrician observes the composite $\tilde{Y}_i=Y_i+\varepsilon_i$ and select the largest $k$ order statistics
\begin{align*}
\tilde{\bY} = ( \tilde{Y}_{(1)},\tilde{Y}_{(2)},\dots ,\tilde{Y}_{(k)}).
\end{align*}
The following proposition shows that $\tilde{\bY}$ has the same limit as $\bY$ in (\ref{eq:limit_dist}).

\begin{proposition}\label{prop:robust}
Suppose
\begin{enumerate*}
\item\label{item:robust1} $Y_i$ is \iid with CDF satisfying the Pareto-type tail condition as in Condition \ref{cond2} with tail index $\xi>0$, and
\item\label{item:robust2} $\varepsilon_i$ is independent from $Y_i$ and satisfies $\sup_i \abs{\varepsilon_i}=o_p(n^\xi)$.
\end{enumerate*}
Then there exist sequences of constants $a_n$ and $b_n$ such that for any fixed $k$,
\begin{equation*}
\frac{\tilde{\bY}-b_n}{a_n}\Rightarrow \bV
\end{equation*}
as $n\to \infty$, where $\bV$ is defined as in (\ref{eq:limit_dist}).
\end{proposition}

\begin{proof}%[Proof of Proposition \ref{prop:robust}]
The proof is analogous to that of Lemma 1 in \citet{sasaki2020testing}. 
By Corollary 1.2.4 and Remark 1.2.7 in \citet{deHaan2006book}, the constants $a_n$ and $b_n$ can be chosen as $a_n=Q_{Y}\left(1-1/n\right)=O(n^{\xi}) $ and $b_n=0$, where $Q_Y(\cdot)$ denotes the quantile function of $Y$.

Now, let $I=(I_{1},\dots ,I_{k})\in \{1,\dots ,n\}^{k}$ be the $k$ random indices such that $Y_{(j)}=Y_{I_j} $, $j=1,\dots ,k$, and let $\tilde{I}$ be
the corresponding indices such that $\tilde{Y}_{( j) } = \tilde{Y}_{\tilde{I}_j} $. 
Then, the convergence of $\tilde{\bY}$ follows from (\ref{eq:limit_dist}) once we establish $\abs{\tilde{Y}_{\tilde{I}_j}-Y_{I_j}}=o_{p}(a_n)$ for $j=1,\dots,k$. 
We present the case of $k=1$, but the argument for a general $k$ is similar. 
Recall that $\varepsilon_i = \tilde{Y}_i - Y_i$, and Condition \ref{item:robust2} in the proposition implies that $\sup_i \abs{\varepsilon}=o_p(a_n)$.

Given this result, we have that, on one hand, 
\begin{equation*}
\tilde{Y}_{\tilde{I}}=\max_i\{ {Y}_i +\varepsilon _i\}\le
Y_{I} +\sup_i \abs{\varepsilon_i} =Y_{I} +o_{p}(a_n).
\end{equation*}
On the other hand,
\begin{align*}
\tilde{Y}_{\tilde{I}}&=\max_i\{Y_i +\varepsilon _i\}\ge \max_i\{Y_i
+\min_i\{\varepsilon _i\}\}\\
&\ge Y_{I}+\min_i\{\varepsilon _i\}\ge 
Y_{I}-\sup_i \abs{\varepsilon_i} = Y_{I}-o_{p}(a_n).
\end{align*}
Therefore $\abs{\tilde{Y}_{\tilde{I}_j}-Y_{I_j}}=o_{p}(a_n)$ holds.
\end{proof}

Proposition \ref{prop:robust} establishes the robustness to measurement error, which is asymptotically dominated.
Note that this result allows for non-identically distributed errors $\varepsilon_i$ as long as its maximum magnitude is smaller than that of the true wealth.

%%%%%%%%%%%%%%%%%%%%%%%%%%%%%%%%%%%%%%%%%%%%%%%%%%%%%%%%%%
\section{Additional empirical results}\label{sec:additional}

This section examines the robustness of our previous results with more control variables, following the literature on wealth inequality. 
See \citet{AlvaredoAtkinsonPikettySaez2013}, \citet{benhabib2018skewed}, and \citet{BenhabibBisinLuo2019} for example. 
In particular, we include housing returns, stock returns, and the top capital gain tax rate as additional controls. 
We collect data on housing values from United States Census Bureau,\footnote{\url{https://www.census.gov/construction/nrs/historical data/index.html}}
the data on stock returns from Amit Goyal's webpage,\footnote{See \url{http://www.hec.unil.ch/agoyal/} for the original data and their description.} 
and the data on top capital gain tax rate from tax foundation's website.\footnote{ \url{https://taxfoundation.org/federal-capital-gains-tax-collections-historical-data}}
For $\Delta t$-year horizon returns, we construct the continuously compounded $\Delta t$-year returns $r_t(\Delta t) = \log P_{t+\Delta t} - \log P_t$, where $P_t$ is the corresponding housing price or S\&P500 index adjusted for the real term. 
As in Section \ref{sec:results}, we also use the five-year average of housing returns, stock returns, and volatility of these controls. 

Table \ref{tab:estimation4} presents the estimation results, which can be summarized as follows. 
First, the top income tax rate is still significant with a very stable magnitude.
Second, all the additional control variables are not significant except for the stock return in models XIII and XV. 
Third, although not reported, we find that the top corporate tax rate\footnote{See Table 24 at \url{https://www.irs.gov/statistics/soi-tax-stats-historical-data-tables}.} is strongly correlated with the top income tax rate (the correlation coefficient is approximately 0.9).
Including both taxes leads to potential multi-collinearity.
Therefore, our results should be interpreted as composite effects of the overall tax rate, not necessarily a particular tax rate.

%%%%%%%%%%%%%%%%%%%%%%%%%%%%%%%%%%%%%%%%%%%%%%%%%%%%%%%%%%%%%%%%%%%%%%
\begin{table}[!htb]
	\centering
	\renewcommand{\arraystretch}{.75}
	%\resizebox{\linewidth}{!}{%
	\scalebox{0.82}{
		\begin{tabular}{lccccccccc}
		\toprule
		                     & (VI) & (XII) & (XIII) & (XIV) & (XV) \\
		\midrule
			Top Tax Rate       &-1.83$^{*}$&-1.85$^{*}$&-1.69&-1.80$^{*}$&-1.70$^{**}$ \\
			                          &(1.04)&(0.98)&(1.03)&(1.02)&(0.86) \\
			\\
			Volatility              & 0.73 &      &   1.62  &    0.61  & 1.57\\
			                          & (1.04) &      &  (0.99) &   (1.11)   &(1.01)\\
			Bankruptcy Rate    & 0.09 &      &   11.47   &   0.08  & 0.08\\
			                           & (5.62) &      & (24.08) & (6.38) &(10.34)\\
			Stock Return        &      &  0.15  & 0.28$^{*}$ &      & 0.38$^{**}$ \\
			                          &      &   (0.12) & (0.15) &      &(0.18) \\
		        Housing Return    &      &   0.09   & 0.06 &      & 0.00 \\
			                          &      &  (0.39)  & (0.49) &      &(0.60)\\               
	     Top Capital Gain Tax Rate&      &      &      &  -0.61  & -1.75 \\
			                           &      &      &      &  (1.36)  &(1.63)\\                 
			\\
			Constant           & 0.86 & 1.03 & 0.39 & 1.02 & 0.88 \\
			                   &(0.46)&(0.37)&(0.57)&(0.58)&(0.43) \\
			\\
			Years Averaged   &    5 &    5 &    5&    5 &    5\\
		\bottomrule
		\end{tabular}
	}
	\caption{Estimation of $\theta_0$ based on the MLE with additional variables. ***p$<$0.01, **p$<$0.05, *p$<$0.10 (except for the constant). Standard errors are reported in parentheses. }${}$
	\label{tab:estimation4}
\end{table}
%%%%%%%%%%%%%%%%%%%%%%%%%%%%%%%%%%%%%%%%%%%%%%%%%%%%%%%%%%%%%%%%%%%%%%

%%%%%%%%%%%%%%%%%%%%%%%%%%%%%%%%%%%%%%%%%%%%%%%%%%%%%%%%%%%%%%%%%%%%%%
\section{Simulation studies}\label{sec:simulation}
%%%%%%%%%%%%%%%%%%%%%%%%%%%%%%%%%%%%%%%%%%%%%%%%%%%%%%%%%%%%%%%%%%%%%%

In this section, we evaluate the finite sample properties of the newly proposed MLE through simulations. 
The data generating process (DGP) is designed as follows: $X_t = \rho X_{t-1} + \sqrt{1-\rho^2} u_t$ where $u_t$ is \iid standard normal. 
We set $\rho = 0.5$.
Conditional on $X_t=x$, $Y_{it}$ for $i=1,\dots,n$ are randomly generated from the following three distributions with tail index $\xi = \Lambda (0.5+x\beta_0)$ and $\beta_0=1$:
\begin{enumerate*}
\item Pareto distribution with minimum size 1,
\item absolute value of the Student-$t$ distribution, and
\item double Pareto-lognormal distribution (dPlN).
\end{enumerate*}
The double Pareto-lognormal distribution is the product of independent double Pareto and lognormal variables such that
%dPlN has been documented to fit well to size distributions of economic variables including income \cite{reed2003}, city size \cite{giesen-zimmermann-suedekum2010}, and consumption \cite{Toda2017MD}. \cite{reed-jorgensen2004} show that a dPlN variable $Y$ can be generated as
\begin{equation*}
Y=\exp(\mu+\sigma Z_1+Z_2\xi-Z_3),
\end{equation*}
where $Z_1,Z_2,Z_3$ are independent and $Z_1\sim \cN(0,1)$ and $Z_2,Z_3\sim \mathrm{Exp}(1)$. 
For parameter values, we set $\mu=0$ and $\sigma=0.5$, which are typical for income data.% \cite{Toda2012JEBO}.

We set $T=36$ as in the Forbes data set and experiment with $n=(10^4,10^5,10^6)$.
The number of top order statistics is $k=(100,200,400)$, all of which are small relative to $n$.
The numbers are based on $M=1000$ simulations.

Table \ref{tab:sim} shows the following summary statistics of the estimates of $\beta_0$:
\begin{enumerate*}
\item Bias: $\frac{1}{M}\sum_{m=1}^M(\widehat{\beta}_m-\beta_0)$, where $m$ indexes simulations and $M=1000$. 
\item RMSE: root mean squared error defined by $\sqrt{\frac{1}{M}\sum_{m=1}^M(\widehat{\beta}_m-\beta_0)^2}$.
\item Coverage: the fraction of simulations for which the true value $\beta_0=1$ falls into the 95\% confidence interval.
\item Length: the average length of confidence intervals across simulations.
\end{enumerate*}

\begin{table}[!htb]
\centering
\renewcommand{\arraystretch}{.75}
\begin{tabular}{crrrrrrrrr}
\toprule
DGP & \multicolumn{3}{c}{Pareto} & \multicolumn{3}{c}{$\abs{t}$} & \multicolumn{3}{c}{dPlN}\\
\midrule
$k$ & 100 & 200 & 400 & 100 & 200 & 400 & 100 & 200 & 400\\
\midrule
$n$&\multicolumn{9}{l}{Bias}\\
$10^4$&0.01&0.01&0.01&0.07&0.06&0.07&0.06&0.07&0.08\\
$10^5$&0.01&0.01&0.00&0.03&0.03&0.02&0.03&0.02&0.03\\
$10^6$&0.02&0.01&0.00&0.03&0.02&0.01&0.01&0.01&0.01\\
\midrule
$n$&\multicolumn{9}{l}{RMSE}\\
$10^4$&0.19&0.13&0.09&0.21&0.15&0.12&0.20&0.15&0.13\\
$10^5$&0.19&0.12&0.09&0.19&0.13&0.10&0.19&0.14&0.10\\
$10^6$&0.18&0.13&0.09&0.19&0.13&0.09&0.19&0.13&0.09\\
\midrule
$n$&\multicolumn{9}{l}{Coverage}\\
$10^4$&0.95&0.95&0.95&0.96&0.94&0.88&0.96&0.92&0.86\\
$10^5$&0.95&0.96&0.96&0.95&0.96&0.94&0.95&0.95&0.95\\
$10^6$&0.96&0.96&0.95&0.95&0.96&0.96&0.96&0.95&0.96\\
\midrule
$n$&\multicolumn{9}{l}{Length}\\
$10^4$&0.72&0.49&0.35&0.75&0.51&0.36&0.75&0.52&0.37\\
$10^5$&0.72&0.50&0.35&0.73&0.51&0.35&0.73&0.51&0.36\\
$10^6$&0.72&0.50&0.35&0.72&0.50&0.35&0.72&0.50&0.35\\

\bottomrule
\end{tabular}
\caption{Finite sample properties of the proposed MLE. Based on 1000 simulation draws. See the main text for more details.}\label{tab:sim}
\end{table}

We can make a few observations from Table \ref{tab:sim}. 
First, when the underlying distribution is exactly Pareto, the bias of the MLE is very small as shown in the first three columns, and the coverage probability is also close to the nominal level 95\%.
Second, however, when the underlying distribution is not Pareto, any parametric estimator that relies on the Pareto tail assumption would suffer from the misspecification bias.
In our case, this bias is reflected in the Student-t and dPlN case with $n=10^4$.
Such bias becomes negligible eventually when $n$ is sufficiently large, as seen in the rows with $n=10^5$ and $10^6$, which are still smaller in magnitudes than the total population in the U.S.
Therefore, we believe that our estimator as well as the extreme value approximation \eqref{fv} perform well.
Finally, the randomness of our estimator does not decrease with $n$ significantly, which is not surprising. 
A larger $n$ only improves the extreme value approximation, while a larger $T$ will substantially reduce the RMSE.
Again because of the potentially large $n$, we are essentially facing the parametric problem where $T$ draws are selected from the extreme value distribution \eqref{fv}, and therefore a $T$ as small as 36 still produces informative and significant estimates.

\singlespacing
\bibliographystyle{plainnat}
\bibliography{bib}
%%%%%%%%%%%%%%%%%%%%%%%%%%%%%%%%%%%%%%%%%%%%%%%%%%%%%%%%%%%%%%%%%%%%%%
%\section{Robustness\label{sec:robustness}}
%%%%%%%%%%%%%%%%%%%%%%%%%%%%%%%%%%%%%%%%%%%%%%%%%%%%%%%%%%%%%%%%%%%%%%

\end{document}